\documentclass[onecolumn,abstracton,titlepage=false,DIV=8,headings=standardclasses]{scrartcl}
\usepackage[utf8]{inputenc}
\usepackage{authblk}

\usepackage[english]{babel}

\usepackage{newtxtext,newtxmath,dsfont,microtype}
\usepackage[format=plain,font={small,it},labelfont=bf,up]{caption}

\usepackage{xspace}

\usepackage{amsmath,amssymb,amsthm,url,graphicx,mathtools}
\usepackage{booktabs,makecell}
\usepackage{array}
\newcolumntype{L}[1]{>{\raggedright\let\newline\\\arraybackslash\hspace{0pt}}m{#1}}
\newcolumntype{C}[1]{>{\centering\let\newline\\\arraybackslash\hspace{0pt}}m{#1}}
\newcolumntype{R}[1]{>{\raggedleft\let\newline\\\arraybackslash\hspace{0pt}}m{#1}}

\usepackage[svgnames]{xcolor}
\usepackage[colorlinks,pdftex,urlcolor=Navy,linkcolor=Crimson,citecolor=Navy]{hyperref}

\usepackage{tikz}
\usetikzlibrary{calc,matrix}

\usepackage{cleveref}
\crefname{lemma}{lemma}{lemmas}
\crefname{proposition}{proposition}{propositions}
\crefname{definition}{definition}{definitions}
\crefname{theorem}{theorem}{theorems}
\crefname{conjecture}{conjecture}{conjectures}
\crefname{corollary}{corollary}{corollaries}
\crefname{example}{example}{examples}
\crefname{section}{section}{sections}
\crefname{appendix}{appendix}{appendices}
\crefname{figure}{fig.}{figs.}
\crefname{equation}{eq.}{eqs.}
\crefname{table}{table}{tables}
\crefname{item}{property}{properties}
\crefname{remark}{remark}{remarks}

\newtheorem{theorem}{Theorem}
\newtheorem{definition}[theorem]{Definition}
\newtheorem{corollary}[theorem]{Corollary}

\newtheorem{lemma}[theorem]{Lemma}

\newtheorem{remark}[theorem]{Remark}
\include{braket}

\DeclareMathOperator{\poly}{poly}

\DeclareMathOperator{\spn}{span}

\DeclareMathOperator{\BigO}{O}

\DeclareMathOperator{\gs}{\mathcal{L}}
\newcommand{\op}{\mathbf}

\renewcommand{\H} {{\ensuremath{\mathcal H}}\xspace}
\newcommand{\1} {\ensuremath{\mathds 1}}

\newcommand{\lham}{\textsc{Local Hamiltonian}\xspace}
\newcommand{\sat}{\textsc{3-sat}\xspace}

\makeatletter
\def\munderbar#1{\underline{\sbox\tw@{$#1$}\dp\tw@\z@\box\tw@}}
\makeatother

\newcommand{\field}[1] {\mathds{#1}}
\renewcommand{\P}{\textnormal{{P}}\xspace}
\newcommand{\NP}{\textnormal{{NP}}\xspace}

\newcommand{\BQP}{\textnormal{{BQP}}\xspace}
\newcommand{\BQPSPACE}{\textnormal{{BQPSPACE}}\xspace}
\newcommand{\PSPACE}{\textnormal{{PSPACE}}\xspace}
\newcommand{\PQPSPACE}{\textnormal{{PQPSPACE}}\xspace}
\newcommand{\BQEXP}{\textnormal{{BQEXP}}\xspace}
\newcommand{\BQEXPSPACE}{\textnormal{{BQEXPSPACE}}\xspace}
\newcommand{\QMA}{\textnormal{{QMA}}\xspace}
\newcommand{\QMAEXP}{\textnormal{{QMA\textsubscript{EXP}}}\xspace}
\newcommand{\StoqMA}{\textnormal{{StoqMA}}\xspace}
\newcommand{\yes}{\textnormal{{YES}}\xspace}
\newcommand{\no}{\textnormal{{NO}}\xspace}

\DeclareMathOperator{\lmin}{\lambda_\mathrm{min}}

\makeatletter
\newcommand{\subalign}[1]{%
  \vcenter{%
    \Let@ \restore@math@cr \default@tag
    \baselineskip\fontdimen10 \scriptfont\tw@
    \advance\baselineskip\fontdimen12 \scriptfont\tw@
    \lineskip\thr@@\fontdimen8 \scriptfont\thr@@
    \lineskiplimit\lineskip
    \ialign{\hfil$\m@th\scriptstyle##$&$\m@th\scriptstyle{}##$\crcr
      #1\crcr
    }%
  }
}
\makeatother

\usepackage[backend=bibtex,style=alphabetic,doi=false,isbn=false,url=false,maxbibnames=20,sorting=none]{biblatex}
\bibliography{../../Bibliography/hamilton}
\bibliography{../../Bibliography/gadgets}
\bibliography{../../Bibliography/books}
\bibliography{../../Bibliography/commuting}

\newbibmacro{string+doi}[1]{\iffieldundef{doi}{#1}{\href{https://dx.doi.org/\thefield{doi}}{#1}}}
\DeclareFieldFormat{title}{\usebibmacro{string+doi}{\mkbibemph{#1}}}
\DeclareFieldFormat[article,thesis,incollection,inproceedings]{title}{\usebibmacro{string+doi}{\mkbibquote{#1}}}

\iffalse
\usepackage{showlabels}
\usepackage[colorinlistoftodos]{todonotes}
\usepackage{soul}

\else
\newcommand{\listoftodos}{}
\fi

\renewenvironment{abstract}
 {\small
    \list{}{
    \setlength{\leftmargin}{0cm}%
    \setlength{\rightmargin}{\leftmargin}%
  }%
  \item\relax}
 {\endlist}
 
 \makeatletter
 \renewcommand{\maketitle}{\bgroup\setlength{\parindent}{0pt}
 \begin{flushleft}
   \textbf{\huge\@title}\\[.5cm]
   \Large\@author
 \end{flushleft}\egroup
 }
 \makeatother

\begin{document}
\thispagestyle{empty}
\listoftodos
\newpage

\title{Perturbation Gadgets:\\Arbitrary Energy Scales\\[4mm]from a Single Strong Interaction}
\author{\textbf{Johannes Bausch}$^1$\\$^1$\normalsize{DAMTP, CQIF, University of Cambridge, \texttt{jkrb2@cam.ac.uk}}}


\maketitle
\thispagestyle{empty}
\enlargethispage{2cm}

\begin{abstract}
\textbf{Abstract.} Fundamentally, it is believed that interactions between physical objects are two-body.
Perturbative gadgets are one way to break up an effective many-body coupling into pairwise interactions: a Hamiltonian with high interaction strength introduces a low-energy space in which the effective theory appears $k$-body and approximates a target Hamiltonian to within precision $\epsilon$.
One caveat of existing constructions is that the interaction strength generally scales exponentially in the locality of the terms to be approximated.

In this work we propose a many-body Hamiltonian construction which introduces only a single separate energy scale of order $\Theta(1/N^{2+\delta})$, for a small parameter $\delta>0$, and  for $N$ terms in the target Hamiltonian.
In its low-energy subspace, we can approximate any normalized target Hamiltonian $\op H_\mathrm{t}=\sum_{i=1}^N \op h_i$
with norm ratios $r=\max_{i,j\in\{1,\ldots,N\}}\|\op h_i\| / \| \op h_j \|=\BigO(\exp(\exp(\poly n)))$ to within \emph{relative} precision $\BigO(N^{-\delta})$.
This comes at the expense of increasing the locality by at most one, and adding an at most poly-sized ancilliary system for each coupling; interactions on the ancilliary system are geometrically local, and can be translationally-invariant.

In order to prove this claim, we borrow a technique from high energy physics---where matter fields obtain effective properties (such as mass) from interactions with an exchange particle---and a tiling Hamiltonian to drop all cross terms at higher expansion orders, which simplifies the analysis of a traditional Feynman-Dyson series expansion.

As an application, we discuss implications for \QMA-hardness of the \lham problem, and argue that ``almost'' translational invariance---defined as arbitrarily small relative variations of the strength of the local terms---is as good as non-translational-invariance in many of the constructions used throughout Hamiltonian complexity theory.
We furthermore show that the choice of geared limit of many-body systems, where e.g.\ width and height of a lattice are taken to infinity in a specific relation, can have different complexity-theoretic implications: even for translationally-invariant models, changing the geared limit can vary the hardness of finding the ground state energy with respect to a given promise gap from computationally trivial, to \QMAEXP-, or even \BQEXPSPACE-complete.
\end{abstract}
\newpage
\tableofcontents

\section{Introduction}
In nature, the way particles can interact is inherently limited.
Just like in a game of billiards, where under high-enough time resolution every ball-to-ball contact can be discriminated in principle, many-body systems are believed to be governed by two-body interactions.
When we relax the time resolution---and for instance only check the billiard table every half second---it appears as if multiple balls have interacted simultaneously, and one can derive an effective multi-body theory from these observations.

While many-body terms appear in real-world systems, e.g.\ in rare-gas liquids \cite{Jakse2000}, where describing thermodynamic properties accurately requires the introduction of a three-body term, to model polar molecules \cite{Bonnes2010} or phases of charged particles in suspension \cite{Wang2018}, their occurence is rare.
For the field of Hamiltonian complexity theory, which tries to link rigorous complexity-theoretic statements like ``how hard is it to estimate the ground state energy of a local Hamiltonian?'' to realistic systems---e.g.\ by requiring low local dimension, a realistic set of interactions, and nearest-neighbour couplings only---this is of course a conundrum: hardness constructions usually work by mapping a type of constraint satisfaction problem to the interactions of a many-body system.
If the interactions get more restricted, the constraints become easier to solve.

In order to circumnavigate this problem, reductions are typically proven in two steps: at first, one allows the freedom of choosing long-range interactions, which makes the task of embedding a hard problem into a local Hamiltonian significantly easier.
As a second step, one uses a technique called \emph{perturbation gadgets} to break down effective $k$-local terms to two-body couplings.

This breaking down of an effective high-locality interaction into two-body couplings is reminiscent of a renormalization group (RG) step, applied in inverse direction, as e.g.\ a block spin RG applied to the Ising model \cite{Kadanoff1966,Verstraete2005}.
In this example, a square grid of spins interacting via nearest-neighbor Ising couplings $J_0 \sigma_i\sigma_j$ at temperature $T_0$ is ``coarse-grained'' into $2\times 2$ blocks, which can then be described via a variant of the original dynamics, but with different parameters $J_1$ and $T_1$.
A single RG step---i.e.\ from four $2^{i-1} \times 2^{i-1}$ blocks to a block of $2^i \times 2^i$ spins---is thus qualitatively similar to a perturbation gadget, where individual spins are grouped  together to produce an effective interaction at a different length- and energy scale as the original couplings \cite{Cao2015}.
Yet a crucial difference is that for RG flow one intends the effective interactions to remain of the same kind, with potentially different parameters---where e.g.\ if the RG is iterated in the Ising example, $J_i$ and $T_i$  approach (potentially infinite) fixed points for $i\longrightarrow\infty$---a perturbation gadget is used to create \emph{more complex} types of interactions \cite{Cubitt2017}.

Effective theories usually introduce a separate energy scale $\Delta$, which has to increase with the system size in order to suppress the introduced errors.
This scaling is usually quite drastic: to break down a $k$-local interaction to $2$-body with an error $\epsilon$, $\Delta$ commonly has to scale like $\Omega(1/\epsilon^k)$, where $\epsilon=1/\poly n$ in the system size $n$.
Yet having a coupling constant which increases as the system grows is highly unphysical---in particular because the typical polynomial degree of $\epsilon^{-1}$ itself is huge, e.g.\ in the context of \QMA-hardness constructions, where $\epsilon$ scales inverse quadratically in the runtime of the computation, which itself can be an arbitrary polynomial in the system size $n$.

In a recent study \cite{Cao2017a}, the authors have analysed how the scaling of $\Delta$ can be improved by an effective numerical algorithm, which yields tighter bounds than suggested by perturbation theory alone. Yet while the bounds are improved by several orders of magnitude, the asymptotic scaling appears to remain unfavourable (see e.g.\ \cite[fig.\ 5]{Cao2017a}).

In this paper, we propose a novel method which allows the introduction of only a single scaling constant with vastly-reduced overhead as compared to the typical $\Delta$ required in a perturbative expansion.
The aim of this work is not to replace gadget constructions, but to augment them: it can be applied to any construction of a Hamiltonian $\op H$ with various energy scales up to relative strength that scales doubly-exponential in the size of the system, i.e.\ $\exp(\exp(\poly n))$.
However, as in the gadget case we cannot get away with no scaling constant at all.
For our construction, a strong interaction with weight $\BigO(N^{2+\delta})$ is necessary to simulate $\op H$ in an effective subspace up to \emph{relative} accuracy $\BigO(N^{-\delta})$, where $N$ is the number of local terms present in the target Hamiltonian.
We emphasize that this approximation is independent of the original scale $\Delta$ one wishes to obtain.
This comes at a cost: the effective Hamiltonian is normalized to $\BigO(1)$, and one has to introduce an ancilliary system for every interaction present in the original construction that features a scaling operator norm.
The ancilliary system  is a geometrically local and translationally-invariant nearest-neighbour spin chain which couples locally to the system at hand.
This means that we need to potentially increase the locality of the original construction by one---where we emphasize that this is \emph{only} necessary if the interaction with scaling norm are already $k$-local for a $k$-local Hamiltonian.\footnote{A counterexample would, for instance, be a $2$-local Hamiltonian with additional $1$-local on-site interactions that vary; as only the latter will have their locality increased by one, the overall Hamiltonian is still $2$-local.}

While it is true that it seems to defeat the purpose of perturbation gadgets to first break down high-locality interactions to two-body, only then to increase them back to three-local, we argue that our construction improves the picture in two aspects.
\begin{enumerate}
\item Our scaling is independent of the locality of the original construction, and thus superior to e.g.\ stopping perturbation theory of a 10-local Hamiltonian once the interactions are 3-local.
\item We introduce a \emph{relative} overall error only. This is particularly useful for hardness constructions, where e.g.\ a small promise gap of $1/\poly n$ has to maintained. For us, a relative error of say $1/10$ would thus suffice.
\end{enumerate}

\begin{table}[t]
\hspace{-2.1cm}
\small
\begin{tabular}{l | llll}
	\toprule
	                                                       & technique              & \makecell{locality $k$\\order $l$} & $\Delta=\Omega(\cdot)$ & extra terms per interaction      \\ \midrule
	\citeauthor{Piddock2015} \cite{Piddock2015,Bravyi2014} & S-W                    & $l\le4$                & $\epsilon^{-l}\|\op V\|^{l(l+1)}$ &  \\
	\citeauthor{Cao2015} \cite{Cao2015}                    & F-D                    & $k=2,3$                & $\BigO(\epsilon)$                 & $\Omega(\epsilon^{-2},\| \op V\|^{2})$-sized cliques$^\dagger$ \\
	\citeauthor{Kempe2006} \cite{Kempe2006}                & F-D                    & $k=3$                  & $\epsilon^{-3}$                   &  3 ancillas\\
	\citeauthor{Bravyi2011} \cite{Bravyi2011}              & S-W                    & $l\in\field N$         & $\epsilon^{-(l+1)}\|\op V\|^{l+1}/l^2$                     &  \\
	\citeauthor{Jordan2008} \cite{Jordan2008}              & Bloch                  & $k\in\field N$         & $^\ddagger$        &  $k$-sized cliques \\
	\citeauthor{Oliveira2008} \cite{Oliveira2008}          & F-D                    & $k\mapsto \lceil k/2\rceil+1$    & $\epsilon^{-2} (\|\op V\| + r)^6$                         & 1 ancilla$^\S$ \\
 \bottomrule
\end{tabular} 
\caption{Examples for perturbation gadgets using various expansion techniques, with required gap scaling, interaction graph modifications, and coupling scaling $\Delta$ in the parameters: approximation error $\epsilon$, operator norm of the target Hamiltonian $\| \op V \|$.\\
$^\dagger$One per 2-body interaction. The paper contains a direct proposal for three-body interactions; for higher-order terms, the authors also propose taking another gadget to break $k$-body to $2$-body, and then reduce the weight with their method.\\
$^\ddagger$The authors show series convergence for $\Delta>\| \op V\|/k$; no analytical error analysis is given.\\
$^\S$For the mediator gadget: $r$ is $\propto\max\{\|\op A\|,\|\op B\| \}$ for the $k$-local interaction term $\op A\otimes\op B$.
}
\end{table}

\newcommand{\targ}{_\mathrm{t}}
\newcommand{\hsim}{_\mathrm{sim}}
The notion of perturbation gadgets is tightly-linked to the idea of simulation of quantum systems.
The theory is well-developed, and we only summarize the central points here; we focus on the simpler definition in \cite{Bravyi2014}, but refer the reader to \cite{Cubitt2017} for an in-depth discussion.
Formally, the ability to simulate (the static properties of) one quantum system with another means that one can reproduce either the eigenvalues, the eigenvectors---or both---of some target Hamiltonian $\op H\targ$ within some invariant subspace $\mathcal L\subset\H\hsim$ (e.g.\ the low-energy subspace) of a simulator Hamiltonian $\op H\hsim$.

Since the Hilbert spaces on which $\op H\targ$ and $\op H\hsim$ are defined---denoted $\H\targ$ and $\H\hsim$---are usually not identical, we need to allow for an \emph{encoding} map $\mathcal E:\H\targ\mapsto\H\hsim$;
then $\H\hsim$ together with $\mathcal E$ \emph{simulate} $\op H\targ$ with error tuple $(\epsilon,\eta)$ if there exists an isometry $\tilde{\mathcal E}:\H\targ\mapsto\H\hsim$ such the image of $\tilde{\mathcal E}$ is $\mathcal L$, and further $\| \op H\targ-\tilde{\mathcal E}^\dagger\op H\hsim\tilde{\mathcal E} \|\le\epsilon$ and $\| \mathcal E - \tilde{\mathcal E} \|\le\eta$.
Roughly speaking, the first two conditions imply that the eigenvalues of $\op H\targ$ are reproduced up to error $\epsilon$; the latter implies closeness of the eigenvectors up to error $\eta$ (see \cite[def.\ 1, lem.\ 1\&2]{Bravyi2014}).
The reason for this distinction is that while the exact mapping $\tilde{\mathcal E}$ might be very complicated and does not tell us anything about the eigenvectors, we can approximate it via an encoding; since the two maps are close in operator norm we can also reach closeness of the eigenvectors with the effective simulated Hamiltonian.

Since our goal is to reproduce the entire target Hamiltonian within a low-energy space of a simulator Hamiltonian, and since we will employ a well-established series expansion, we will generally disregard the explicit distinction between $\epsilon$ and $\eta$; the self-expansion theorems in \cref{sec:self-energy} capture the two notions of approximation that suffice for our purposes.

\section{Preliminaries}
A Hamiltonian is a hermitian operator $\op H$ on a finite dimensional Hilbert space $\H$.
We say $\op H$ is $n$-body if $\H=(\field C^d)^{\otimes n}$ for some $n,d\in\field N$.
The Hamiltonian $\op H$ is $k$-local if $\op H=\sum_{i=1}^n \op h_i$, and such that $\op h_i$ are Hermitian matrices that each act non-trivially only on $k$ of the subsystems of $\H$.
More precisely, we demand that $\op h_i = \op q_{i,S_i} \otimes \1_{S_i^c}$, where $\op q_{i,S}$ is a Hermitian operator on a subset $S_i \subset \{1,\ldots,n\}$ of size $|S_i|\le k$, and $\1_{S_i^c}$ the identity operation on the complement of $S_i$.
We also call the $\op h_i$ local coupling or interaction terms, and if $\op h_i$ is part of a $k$-local Hamiltonian $\op H$, then $\op h_i$ is---in itself---an at most $k$-body interaction.
Indeed, as mentioned in the abstract, fundamentally physical systems are believed to be interacting via two-body interactions, which means that the Hamiltonian describing such systems is two-local.

If there is a topological structure associated to the Hilbert space $\H$---e.g.\ if each of the $d$-dimensional spaces is associated to the vertices of a graph---then we speak of $\op H$ being \emph{geometrically local} if the local interaction terms $\op h_i$ act in a local fashion with respect to this topology, which usually means that the $k$ vertices that $\op h_i$ acts on have to be connected.
For instance, if the $\op h_i$ are interaction terms between neighbouring $d$-dimensional spins on a grid of side length $L\times L$ (each spin with Hilbert space $\field C^d$, which we also call a d-dimensional qudit), then $\op H$ is a 2-local, $L^2$-body, nearest-neighbour Hamiltonian on a square lattice.

If the topology permits and is e.g.\ like a hyperlattice, we can speak of translational invariance, which means that for all the local terms $\op q_{i,S_i} = \op q_{S_i}$, and $\op H=\sum_{i=1}^n$ is such that the interactions on the underlying graph are invariant under translations---modulo boundary effects; for translationally invariant systems we generally assume open boundary conditions.

The interaction degree of a Hamiltonian is then the maximum number of local terms $\op h_i$ acting non-trivially on any site; it coincides with the degree of the graph describing the interaction topology of $\op H$.
A Hamiltonian with fixed interaction degree then has an interaction degree $\le D$ for some $D\in\field N$, which we keep implicit.
Similarly, we will often leave the locality unspecified when speaking of local Hamiltonians, which simply implies that that the Hamiltonian is $k$-local for some constant $k$.

\subsection{Feynman-Dyson Series}\label{sec:f-d}
Because a lot of our construction hinges on employing a well-known series expansion---the Feynman-Dyson series---and to introduce the notation used throughout the rest of the paper, we will spend some time explaining how to approximate low energy spectra of a sum of a Hamiltonian $\op H$ and a perturbation $\op V$.
We follow the excellent and more thorough introductions within \cite{Kempe2006,Piddock2015}.

Assume we are given a Hamiltonian $\tilde{\op H}:=\op H+\op V$, where $\op H$ has a spectral gap $\Delta$ above its ground space $\gs(\op H)$.
We further assume that $\|\op V\|<\Delta/2$.

\paragraph{Notation.}
Denote the eigenvalues and eigenvectors of $\op H$ ($\tilde{\op H}$) with $\lambda_i$ and $\ket{\psi_i}$ ($\tilde\lambda_i$ and $\ket*{\tilde\psi_i}$), such that $\lmin(\op H)=:\lambda_0$ is the ground state of $\op H$.
Let $\lambda^*:=\lmin(\op H) + \Delta/2$ midway within the spectral gap of $\op H$, and let $\Pi_-$ be the projector onto $\gs(\op H)$---and $\Pi_+$ onto its orthogonal complement, respectively.
We define the \emph{resolvent} of $\op H$ via
\begin{equation}\label{eq:resolvent}
    \op G(z):=(z\1-\op H)^{-1}=\sum_{i}(z-\lambda_i)^{-1}\ketbra{\psi_i},
\end{equation}
and analogously $\tilde{\op G}(z)$ for $\tilde{\op H}$; we note that both resolvents have first order poles at $z=\lambda_i$ or $z=\tilde\lambda_i$, respectively.
The \emph{self-energy} of $\op H$ is then given by
\begin{equation}\label{eq:self-energy}
    \Sigma_-(z):=z\1_--\left[\tilde{\op G}^{-1}(z)\right]_-,
\end{equation}
where the subscripts on an operator $\op A$ are defined via the restriction to the support of the projections $\Pi_\pm$, e.g.\ $\op A_-:=\op A|_{\gs(\op H)}$ (such that $\1_-$ denotes the identity on $\gs(\op A)$), and analogously $\op A_+$ is the restriction to the complement of $\gs(\op H)$.
We will also use the mixed subscripts, best defined in a representation of the Hilbert space $\gs(\op H) \oplus \gs(\op H)^\perp$, where the operator $\op A$ block-decomposes as
\[
\op A = \begin{pmatrix}
\op A_+ & \op A_{+-} \\ \op A_{-+} & \op A_-
\end{pmatrix}.
\]
This also means that the order of operations in \cref{eq:self-energy}---restriction to the low-energy subspace and operator inversion---is irrelelevant for all $z\not\in \{ \lambda_i \}$, i.e.\ where $\tilde{\op G}(z)$ is invertible; for simplicity of notation we thus drop the brackets where appropriate, and identify e.g.\ $\tilde{\op G}^{-1}_-(z) := \left[ \tilde{\op G}^{-1}(z) \right]_-$ (which is thus nothing but $z\1_- - \tilde{\op H}_-$).

If we solve \cref{eq:self-energy} via $\Sigma_-(z) = z\1_- - \left[(z\1_- -\tilde{\op H}_-)^{-1}\right]^{-1} = \tilde{\op H}_-$,
we see that the self-energy $\Sigma_-(z)$ is nothing but the low-energy part of the Hamiltonian $\tilde{\op H}$---where it is important to note that ``low-energy'' in this context means with respect to the spectrum of the unperturbed Hamiltonian $\op H$, \emph{not} $\tilde{\op H}$.
This is not useful per se, though; we do not know how to calculate the effective low-energy Hamiltonian of $\tilde{\op H}$.
On the other hand, we can use a series expansion to approximate it, starting from $\Sigma_-(z)$.
Since $\op G_{+-}^{-1}(z)=\op G_{-+}^{-1}(z)=0$ by construction, note
\begin{align*}
    \tilde{\op G}(z)
    &=(z\1 - \tilde{\op H})^{-1} = (z\1 - \op H - \op V)^{-1}=(\op G^{-1}(z)-\op V)^{-1} \\
    &=\begin{pmatrix}
    \op G_+^{-1}(z) - \op V_+ & -\op V_{+-} \\
    -\op V_{-+} & \op G_-^{-1}(z)-\op V_-
    \end{pmatrix}^{-1}
    =:\begin{pmatrix}
    \op A & \op B \\ \op C & \op D
    \end{pmatrix}^{-1}.    
\end{align*}
The lower-right block of $\tilde{\op G}(z)$ is then given by the Schur complement
\[
    \tilde{\op G}_-^{-1}(z)=\op D-\op C\op A^{-1}\op B  =
    \op G^{-1}_-(z)-\op V_- -\op V_{-+}(\op G_+^{-1}(z)-\op V_+)^{-1}\op V_{+-}.
\]
Dropping the argument $z$ in $\op G_+=\op G_+(z)$ for brevity, we further have
\begin{align*}
    (\op G_+^{-1}-\op V_+)^{-1}
    &=(\op G_+^{-1}(\1_+ -\op G_+\op V_+))^{-1}=(\1_+ -\op G_+\op V_+)^{-1}\op G_+\\
    &=\op G_+ + \op G_+\op V_+\op G_+ + \op G_+\op V_+\op G_+\op V_+\op G_+ + \ldots
\end{align*}
as a geometric series expansion, which converges if $\| \op G_+\op V_+\|<1$.
Under this assumption, we can conclude
\begin{equation}\label{eq:Z-expansion}
    \Sigma_-(z)=\op H_- + \op V_- + \op V_{-+}\op G_+\op V_{+-} + \op V_{-+}\op G_+\op V_+\op G_+\op V_{+-} + \ldots.
\end{equation}

\subsection{Self-Energy Expansion Theorems}\label{sec:self-energy}
There is two major variants of approximations that can result from this self-expansion using the Feynman-Dyson series.
Representative of the literature we quote the following two variants.

\newcommand{\eff}{_\mathrm{eff}}
\begin{theorem}[\citeauthor{Cao2017a} \cite{Cao2017a}]\label{th:pert-1}
Let $\tilde{\op H}=\op H+\op V$ as above, and assume $\|\op V\|\le\Delta/2$.
Let $\epsilon>0$.
If there exists a Hamiltonian $\op H\eff$ with spectrum $\{\lambda_1,\ldots,\lambda_k\}$ contained in an interval $[a,b]$, $a<b<\Delta/2-\epsilon$,
and for all $z\in[a-\epsilon,b+\epsilon]$ it holds that $\|\Sigma_-(z)-\op H\eff\|\le\epsilon$,
then each $\lambda_i$ is $\epsilon$-close to the $i$\textsuperscript{th} eigenvalue of $\tilde{\op H}_-$.
\end{theorem}

Note that in general we will have a dependence $\epsilon=\epsilon(\Delta)$; however, if we only request that the error be small, but not shrinking with the system size, we can keep the ratio of the terms $\op H$ and $\op V$ fixed.
The following variant allows one to make a statement not only about the eigenenergies, but also about the eigenvectors.
\begin{theorem}[\citeauthor{Oliveira2008} \cite{Oliveira2008}]\label{th:pert-2}
Let the setup be as in \cref{th:pert-1}, and denote with $\lambda_\mp$ the ground- and first excited energy of $\op H$, respectively.
Let $z_0=(b+a)/2$, $w\eff=(b-a)/2$, and $r$ be the radius of a disc $D$ centered around $z_0$ encompassing the point $b+\epsilon$.
If for all $z\in D$ we have $\|\Sigma_-(z) - \op H\eff\|\le\epsilon$, then
\[
    \| \tilde{\op H}_- - \op H\eff \| \le \frac{3(\|\op H\eff\|+\epsilon)\|\op V\|}{\lambda_+ - \|\op H\eff\|-\epsilon}
    + \frac{r(r+z_0)\epsilon}{(r-w\eff)(r-w\eff-\epsilon)}.
\]
\end{theorem}

In particular, while \cref{th:pert-1} allows us to make a statement about the eigenenergies without requiring $\Delta/\|\op V\|\rightarrow\infty$---which manifests in a constant approximation error for the eigenvectors of $\op H\eff$---\emph{with} said condition and \cref{th:pert-2} we can also approximate the full spectrum of $\op H\eff$ to arbitrary precision.

\subsection{A Bound State Hamiltonian}\label{sec:H-bound}
We will need a variant of a random walk Hamiltonian, used ubiquitously in \QMA-hardness constructions in the context of Feynman's History State construction.
In particular, what we aim to achieve is to create a Hamiltonian on a multipartite Hilbert space, with a constant spectral gap above a unique ground state, and such that the latter has most of its weight localized around a particular site.
Like this, we can ``condition'' an interaction on the ground state away from its localization site.
The intuition is taken from particle physics:
interactions are commonly coupled to an exchange gauge particle;
this coupling is weak when conditioned on a field away from where the gauge particle mostly lives---e.g.\ a photon, whose field drops off away from an electron, influences how strong an electron-electron scattering is depending on how far apart the two electrons are.

\newcommand{\Hb}{\op H_\textnormal{b}}
Let us make this precise.
Let $b>0$.
For an integer $T\ge 2$, let $\Hb$ be a Hamiltonian on $\field C^T$ defined via
\begin{equation}\label{eq:H-bound}
    \Hb :=  - b\ketbra{1} + \sum_{t=1}^{T-1}(\ket t-\ket{t+1})(\bra t-\bra{t+1}),
\end{equation}
where the $\ket t$ label a fixed orthonormal basis.
The second term in \cref{eq:H-bound} is a path graph Laplacian, whereas the first term assigns a bonus term of strength $b$ to the state $\ket1$.

\begin{lemma}\label{lem:H-bound-gs-1}
    For $b>0$, $\Hb$ as defined in \cref{eq:H-bound} has a single ground state with eigenvalue $\lmin<-b^2/(b+1)$.
    All other eigenvalues are positive.
\end{lemma}
\begin{proof}
    Uniqueness of a single negative eigenvalue is a standard argument: assume this is not the case. Then there exist at least two orthogonal eigenvectors $\ket u,\ket v$ with negative eigenvalues, and any $\ket x\in\spn\{\ket u,\ket v\}$ satisfies $\bra x\Hb \ket x<0$.  Since $\dim\ker\ketbra 1=T-1$,
    there exists a nonzero $\ket x\in\spn\{\ket u,\ket v\}$ such that $\ketbra 1\ket x=0$. Therefore $0>\bra x\Hb\ket x=\bra x(\Hb+b\ketbra1)\ket x$, contradiction, since $\Hb + b\ketbra1$ is a path graph Laplacian, which is positive semi-definite.
    
    We make an ansatz for the ground state. Let
    \begin{equation}\label{eq:H-bound-gs}
        \ket{\Psi}:=A\sum_{t=1}^T(b+1)^{-t}\ket t
        \quad\text{where }
        A^2=\frac{b(2+b)}{1-(b+1)^{-2T}}
        \text{\ \ for normalization,}
    \end{equation}
    for which we note $A\in(0,b+1)\ \forall b>0, T\ge2$. Then
    \begin{align}
        \Hb\ket\Psi &= A\sum_{t=1}^T\ket t \times \begin{cases}
        	-(b+1)^{-2} + (b+1)^{-1} - \frac{b}{b+1} & t=1   \\
        	-(b+1)^{-t-1} + 2(b+1)^{-t} - (b+1)^{-t+1}   & 1<t<T  \\
        	-(b+1)^{-T+1} + (b+1)^{-T}                     & t=T
        \end{cases} \nonumber\\
        &= -A\sum_{t=1}^{T-1}\frac{b^2}{(b+1)^{t+1}}\ket t - \frac{Ab}{(b+1)^T}\ket T \nonumber\\
        &= -\frac{b^2}{b+1}\ket\Psi - \frac{Ab}{(b+1)^{T+1}}\ket T.
        \label{eq:Hb-Psi}
    \end{align}
    Thus
    \[
        \bra\Psi \Hb \ket\Psi
        =-\frac{b^2}{b+1} - \frac{A^2b}{(b+1)^{2T+1}} < -\frac{b^2}{b+1}.\qedhere
    \]
\end{proof}
\begin{lemma}\label{lem:H-bound-gs-2}
    We pick $b\ge1$.
    $\Hb$ then has ground state $\ket{\Psi_0}=\ket\Psi + \epsilon\ket{\xi}$, where $\ket\Psi$ is from \cref{eq:H-bound-gs}, $\ket{\xi}$ is normalized, and $\epsilon=\BigO(b\sqrt T/(b+1)^{T})$ where the $\BigO$ limit is taken with respect to $T\longrightarrow\infty$.
\end{lemma}
\begin{proof}
    By absorbing complex phases, choose the eigenvectors $\{\ket{\Psi_i} \}_{i=0}^{T-1}$ of $\Hb$---with ground state $\ket{\Psi_0}$---such that we can represent the ansatz state $\ket{\Psi}=\sum_{i=0}^{T-1} \alpha_i\ket{\Psi_i}$ with $\alpha_i\ge0$ for all $i$.
    By \cref{lem:H-bound-gs-1}, the lowest eigenvalue $\lambda_0 = \lmin(\Hb)\in(-\infty,-b^2/(b+1)]$; all other eigenvalues of $\Hb$ satisfy $\lambda_i \in [0,\infty)$.
    Therefore, for any $s\in(0,b^2/b+1)$,
    \begin{align*}
        \alpha_i 
        &= \braket{\Psi_i}{\Psi}
        = \frac{1}{\lambda_i+s}\bra{\Psi_i}\Hb+s\1\ket{\Psi} \\
        &= -\frac{1}{\lambda_i+s}\left[\left( \frac{b^2}{b+1} - s \right)\braket{\Psi_i}{\Psi} + \frac{Ab}{(b+1)^{T+1}}\braket{\Psi_i}{T}\right],
    \end{align*}
    where in the first line we used the fact that the $\ket{\Psi_i}$ are an orthonormal set of vectors, and in the second line we used the expression of $\Hb\ket{\Psi}$ from \cref{eq:Hb-Psi}.
    Since $b\ge1$, we can choose $s=1/4$. We further have $A\le b+1$. For $i>0$, we know that $\lambda_i\ge0$, and we conclude
    \[
        \alpha_i = \left(
            1+\frac{1}{\lambda_i+\frac14}\left(\frac{b^2}{b+1}-\frac14\right)\right)^{-1}
            \frac{1}{\lambda_i+\frac14}\frac{Ab}{(b+1)^{T+1}} |\braket{\Psi_i}{T}|
        \le \frac{4b}{(b+1)^{T}}.
    \]
    Then
    \[
        | \braket{\Psi_0}{\Psi} |^2 = \alpha_0^2 = 1 - \sum_{i=1}^{T-1}\alpha_i^2 = 1 + \frac{4(T-1)b^2}{(b+1)^{2T}} = \BigO\left(\frac{T b^2}{(b+1)^{2T}}\right)
    \]
    for large $T$, and the claim follows.
\end{proof}

This allows us to approximate to very high precision the amplitudes of the ground- and higher excited states; of particular interest will be the amplitudes for the basis states $\ket{T'}$ for $T' < T$;
the reason for this is that the approximation error in \cref{lem:H-bound-gs-2} (i.e.\ the precision to which we know the ground state at all) is of the same order of magnitude as the smallest amplidude in the ground state, $| \braket{\Psi_0}{T} |$.
However, since we want to be able to accurately fine-tune a specific amplitude of $\ket{\Psi_0}$, we need the corresponding error of that entry to be much smaller.
In order to formalize this notion, we will assume the path graph underlying the graph Laplacian in the definition of $\Hb$ in \cref{eq:H-bound} has a multiple of the original length $T$; we call this multiple $M\in\field N$, $M>1$ throughout the paper, and the target amplitude we wish to estimate and tune remains $\braket{\Psi_0}{T}$.
This is captured in the following corollary.
\begin{corollary}\label{cor:H-bound-gs}
    Let $M\in\field N$, $M>1$, and $b\ge1$. Let $\ket{\Psi_0}$ be the ground state of $\Hb$ on a chain of length $MT$.
    Then
    \[
        |\braket{\Psi_0}{T}|^2=\frac{b (b+2)}{(b+1)^{2 T}} + \BigO\left(\frac{1}{(b+1)^{MT}}\right),
    \]
    where the $\BigO$ limit is taken with respect to $T\longrightarrow\infty$.
\end{corollary}
\begin{proof}
By \cref{lem:H-bound-gs-2},
    \[
        |\braket{\Psi_0}{T}|^2 = | \braket{\Psi}{T} + \epsilon\braket{\xi}{T} |^2 \le | \braket{\Psi}{T} |^2 + 2\epsilon| \braket{\Psi}{T} | + \epsilon^2.
    \]
    First note that by \cref{eq:H-bound-gs}, $\braket{\Psi}{T} = A / (b+1)^T$, where $A$ is the normalization constant defined on a path of length $MT$ (\emph{not} T), such that
    \begin{align*}
        | \braket{\Psi}{T} |^2 
        &= \frac{b(2+b)}{1-(b+1)^{-2MT}} \times \frac{1}{(b+1)^{2T}} = \frac{b(2+b)}{(b+1)^{2T} - (b+1)^{-2MT+2T}} \\
        &= \frac{b(2+b)}{(b+1)^{2T}} + \BigO\left( \frac{1}{(b+1)^{4T}}\times (b+1)^{-2MT+2T}\right) \\
        &= \frac{b(2+b)}{(b+1)^{2T}} + \BigO\left( \frac{1}{(b+1)^{2MT+2T}}\right)
    \end{align*}
    for the $\BigO$-limit taken with respect to $T\longrightarrow\infty$.
    Using the expansion $\sqrt{a+x} = \sqrt{a} + \BigO(x/\sqrt a)$ for $0<x<a$ and the small $x$ limit, we therefore have
    \[
        | \braket{\Psi}{T} | = \frac{\sqrt{b(2+b)}}{(b+1)^T} + \BigO\left( \frac{1}{(b+1)^{MT}} \right).
    \]
    By \cref{lem:H-bound-gs-2} we further have $\epsilon = \BigO( b\sqrt{MT}/(b+1)^{MT} )$ and thus
    \[
        \epsilon| \braket{\Psi}{T} | = \BigO\left( \frac{\sqrt{MT}}{(b+1)^{MT}} \right) \times \BigO\left( \frac{1}{(b+1)^T} \right) = \BigO\left(\frac{1}{(b+1)^{MT}} \right),
    \]
    as $\sqrt{MT} / (b+1)^T \longrightarrow 0$ for $T\longrightarrow\infty$.
    A similar argument bounds $\epsilon^2$; the claim follows.
\end{proof}

Note that e.g.\ choosing $M=4$ suffices such that $| \braket{\Psi_0}{T} |^2$ in \cref{cor:H-bound-gs} equals $b(b+2)/(b+1)^{2T}$ up to a relative factor of $\BigO(1/(b+1)^{2T})$, as intended; it is clear that a tighter error bound can be achieved by increasing $M$ further.
Furthermore, the overlap with a site $T'<T$ is larger; it is therefore possible to expand \cref{cor:H-bound-gs} to obtain the following claim.
\begin{corollary}\label{cor:H-bound-gs-ex}
Let $M\in\field N$, $M>1$, and $b\ge 1$.
On a chain of length $MT$ and for any $T'\le T$, the ground state overlap
\[
|\braket{\Psi_0}{T'}|^2 = \frac{b (b+2)}{(b+1)^{2 T'}} + \BigO\left( \frac{1}{(b+1)^{MT}} \right) 
\]
in the $\BigO$-limit $T\longrightarrow\infty$.
\end{corollary}

In the same fashion as in \cref{cor:H-bound-gs-ex}, we can now immediately deduce the overlap of a state $\ket{T'}$ with the rest of the spectrum of $\Hb$.
\begin{corollary}\label{cor:H-bound-T-overlap}
    Let $M\in\field N$, $M>1$, and $b\ge 1$.
    We consider a chain of length $MT$, and let the eigenstates $\ket{\Psi_i}$ of $\Hb$ be as in \cref{lem:H-bound-gs-2}.
    Then for all $T'\le T$ and in the limit $T\longrightarrow\infty$, we have
    \[
        \sum_{i=1}^{MT-1}|\braket{\Psi_i}{T'}|^2 = 1+\BigO\left(\frac{b(b+2)}{(b+1)^{2T'}} \right)
    \]
\end{corollary}

As we have seen, there is an exponential falloff of the ground state of $\Hb$ away from its bonus term, and the magnitude of overlap $|\braket{\Psi_0}{T}|$ is tightly-controlled by \cref{cor:H-bound-gs,cor:H-bound-gs-ex,cor:H-bound-T-overlap}.
Since $T$ is discrete and we want $b$ to be taken from a fixed interval, an obvious question that arieses is which values $r:=|\braket{\Psi_0}{T'}|^2 \in \field R$ we can construct, by choosing $T$, $T'$, $M$ and $b$ appropriately.
This is a straightforward calculation; yet since we will be interested of the scaling of the parameters $T$, $M$ and $b$ with respect to $r$ we state the result here explicitly.
\begin{lemma}\label{lem:any-coupling-strength}
Let  $r\in(0,1/100)$.
Then there exist an $M\in\field N$, $M>3$, an integer $T\in[\ln(3/r)/\ln 4, \ln(15/r)/\ln 16]$ and a real number $b\in[1,3]$ such that, if $\ket{\Psi_0}$ denotes the ground state of $\Hb$ describing a chain of length $MT$, we have $|\braket{\Psi_0}{T}|^2=r$.
\end{lemma}
\begin{proof}
By \cref{cor:H-bound-gs}, a short calculation yields
\[
    |\braket{\Psi_0}{T}|^2 = r
    \quad\Longleftrightarrow\quad
    T = \frac{\ln(b(b+2)/r)}{2\ln(b+1)} + \epsilon
    \quad\text{for some\ }
    \epsilon = \BigO\left( \frac{1}{(b+1)^{(M-2)T}} \right).
\]
What remains to be shown is that we chan choose $M$ large enough such that for any $r\in(0,1/100)$, there exists a $b\in[1,3]$ such that the above equation is satisfied, even under the restriction that $T$ can only assume an \emph{integer} value.

To prove this, we note that both enumerator and denominator in the expression for $T$ increase monotonically with $b$; their extreme points are thus reached at the endpoints of the interval $b\in[1,3]$.
For the enumerator they are $\ln(3/r)$ and $\ln(15/r)$, for the denominator $2\ln{2}$ and $2\ln{4}$.
We note that the achievable difference $T|_{b=1} - T|_{b=3} = \ln(3/5r)/\ln 16 > 5/4\ \forall r\in(0,1/100)$.
The claim of the lemma then follows from the intermediate value theorem and choosing $M$ large enough such that $\epsilon<1/10$.
\end{proof}
We emphasize that in \cref{lem:any-coupling-strength} we can pick $b, M$ and $T$ such that $| \braket{\Psi_0}{T} |^2=r$ \emph{exactly}, without any remaining error term.
By \cref{cor:H-bound-gs-ex}, we can alternatively demand that $T$ be fixed, and choose to tune the overlap $| \braket{\Psi_0}{T'} |^2$ for some $T'<T$.
Interestingly, if we have multiple copies of the spin Hamiltonian $\op H_{b_i}$, we can achieve the same feat, even under the condition that $M$ and $T$ is identical for all of them.
More precisely, for a range of target overlaps $r_i \in (0,1/100)$, we wish to find states $T_i \le T$ and biases $b_i \in [1,3]$,
such that $r_i = |\braket*{\Psi_{0,i}}{T_i}|^2$ (where $\ket*{\Psi_{0,i}}$ denotes the ground state of $\op H_{b_i}$).
\begin{corollary}\label{cor:any-coupling-strength}
Take a family $\{ r_i \}_{i\in I}$ for a finite index set $I$, such that $r_i\in(0,1/100)\ \forall i$.
Then there exist $M,T\in\field N$, $M>3$ and $T\in[\ln(3/\bar r)/\ln 4, \ln(15/\bar r)/\ln 16]$ where $\bar r:=\min_{i\in I} \{ r_i\}$, and a family of Hamiltonians $\{ \op H_{b_i} \}_{i\in I}$, each on chain length $MT$ and such that for all $i$ there exists a bias $b_i\in[1,3]$ and state $T_i\le T$ such that $|\braket*{\Psi_{0,i}}{T_i}|^2=r_i$, where $\ket*{\Psi_{0,i}}$ is the ground state of $\op H_{b_i}$.
\end{corollary}
\begin{proof}
Follows analogous to \cref{lem:any-coupling-strength}, using \cref{cor:H-bound-gs-ex} instead of \cref{cor:H-bound-gs}.
\end{proof}

For now, this $\Hb$ as defined in \cref{eq:H-bound} acts on a single qudit of dimension $T$; but by the following remark we can ensure the interactions are all defined on a constant local dimension.
\newcommand{\Tlegal}{S_\mathrm{good}}
\begin{remark}\label{rem:clock-times}
Let $\H:=(\field C^d)^s$ be a spin chain of length $s$ and local dimension $d$.
Then the following exists:
Basis states $\{ \ket i \}$ of $\H$ such that $\{ \ket i \} =: \Tlegal \dot{\cup} \Tlegal^c$, where $T:=|\Tlegal|$; define $\Hb'$ on the basis states $\ket t\in\Tlegal$ as in \cref{eq:H-bound}.
Then
\begin{enumerate}
\item $\Hb'$ has only translationally-invariant nearest-neighbour interactions.
\item There exists a $2$-body interaction term $\op p$, such that $\op H := \Hb + \sum_{i=1}^{s-1} \op p_{i,i+1}$---where $\op p_{i,i+1}$ acts on the neighbouring spins $(i,i+1)$ only---such that $\op H$ is block-diagonal with respect to the partition $\Tlegal \cup \Tlegal^c$.
$\op H|_{\spn(\Tlegal )} \cong \Hb$ (unitary equivalence), where $\Hb$ is defined in \cref{eq:H-bound}, but on Hilbert space $\field C^T$.
The other block of $\op H$ satisfies $\op H|_{\spn( \Tlegal^c)}\ge 0$.
\item Either $T=(d-1)\times(s-1)$, or $T=B^{s-3}$ for $B=\lfloor(d-5)/2\rfloor$.
\end{enumerate}
\end{remark}
\begin{proof}
While the proof of this remark is non-trivial---it forms the foundation of Kitaev's seminal proof of QMA-hardness of approximating ground states of local Hamiltonians, see \cite{Kitaev2002} where a $5$-local variant is proven---it has been refined and repeated many times throughout literature (\cite{Kempe2006,Aharonov2009,Oliveira2008,Gottesman2009,Bausch2016,Bausch2017,Nagaj2008,Nagaj2012,Caha2017,Bausch2016a}, amongst others), so we will omit it.
The specific scaling of $T$ with respect to the local dimension $d$ and chain length $s$ can be found for $d=3$ in \cite[Sec.~8.3.4]{Bausch2016}, and $B=6$ in \cite[Sec.~8.3.3]{Bausch2016}; the general $d$ and $B$ cases are immediate consequences, see \cite[Rem.~12]{Bausch2018b}.
\end{proof}
In particular, \cref{rem:clock-times} shows that we can construct translationally-invariant version of the bound state Hamiltonian $\Hb$ from \cref{sec:H-bound}, which has local nearest-neighbour coupling terms, the same single negative-energy ground state $\ket{\Psi_0}$ with weights constrained as e.g.\ in \cref{lem:any-coupling-strength}, and a spectral gap of $\ge 1/2$.

\section{Main Result}
To make rigorous what we mean by one Hamiltonian to approximate another in its low energy subspace, we phrase the following definition.
\begin{definition}\label{def:approximate}
Let $\op H_0$ be a local Hamiltonian on a Hilbert space $\H=(\field C^2)^{\otimes n}$ such that each local term has operator norm bounded by $r(n)$.
We say that $\op H'$ on $\H\otimes\H_2$ approximates $\op H_0$---to error $\epsilon$---in its low-energy subspace if the following conditions hold.
\begin{enumerate}
\item $\op H'$ has a band gap, i.e.\ its spectrum $\sigma(\op H')\subset(-\infty,a) \cup (b,\infty)$ with $a<b$ independent of $n$.
\item Let $\Pi_-$ be the projector onto the lower part of the spectrum, i.e.\ on $\sigma(\op H') \cap (-\infty,a)$.
Then there exists a state $\ket{\psi_0}\in\H_2$ such that
\[
    r(n) \Pi_-\op H'\Pi_- = \op H_0\otimes \ketbra{\psi_0} + \BigO(\epsilon),
\]
where Landau $\BigO(\epsilon)$ term is measured with respect to the operator norm.
\end{enumerate}
\end{definition}

\newcommand{\cl}{_\textnormal{clock}}
\newcommand{\tile}{_\textnormal{tile}}
\begin{theorem}\label{th:main}
Let $\{ \op H_0(n) \}_{n\in\field N}$ be a fixed interaction degree $k-$local family of Hamiltonians, where $\op H_0(n)=\sum_{i=1}^N \op h_i$ is defined on a multipartite Hilbert space $\H=(\field C^d)^{\otimes n}$,
 and where all $N=\poly n$ interactions have norm $\|\op h_i\| = r_i$, where $r_i=r_i(n)$ with $|r_i(n)/r_j(n)|\le r(n)\ \forall i,j$.
Let $\delta>0$.
Then there exists a family of fixed interaction degree $k+1$-local Hamiltonians $\{ \op H'(n) \}_{n\in\field N}$, where $\op H'=\sum_{i=1}^{N'}\op q_i$ on $\H':=\H \otimes \H_2$, $N'=\poly n$, $\H_2=(\field C^q)^{\otimes\poly n}$,
where $1\le\|\op q_i\|\le N^{2+\delta}$, and such that $\op H'(n)$ approximates $\op H_0(n)$ in its low-energy subspace, in the sense of \cref{def:approximate}, with relative error $\BigO(N^{-\delta})$.
The local dimension of the ancilliary system satisfies
\begin{enumerate}
\item $q=3$ if $r=\BigO(\exp(\poly n))$, or otherwise
\item $q=9$ if $r=\BigO(\exp(\exp(\poly n)))$.
\end{enumerate}
\end{theorem}
\noindent
We give a constructive proof of \cref{th:main}; we note that while a variant of \cref{th:main} may in principle also hold for an $r(n)$ that grows faster than doubly exponentially in $n$, our proof does not easily extend to that case.
The next few sections will be spent introducing the machinery necessary for the proof.
As a first step we will prove a slightly weaker variant, where we increase the locality of the interactions by $2$ instead of $1$.
This will save us some tedious algebra in due course, but we will lift the extra constraints and obtain \cref{th:main} in \cref{sec:proof}.

To further simplify notation, we will generally speak of a Hamiltonian $\op H_0$ instead of a family of Hamiltonians $\{ \op H_0(n) \}_{n\in\field N}$---which is the only type of family of Hamiltonians we will be considering here, as per \cref{th:main}; therefore the indexing variable $n$---i.e.\ the system size---will always be clear from the context.

Let for now $\H_2=\H\cl \otimes \H\tile$, where each Hilbert space will be used for one specific step in the construction.
Without loss of generality, we will also assume that the system does not decompose into mutually non-interacting subsets;
if this is the case, we can always regard each system separately.
We first list the two ingredients for our construction.

\subsection{Local Bound State Hamiltonians with Controlled Falloff}
Let $M>3$ be a fixed integer.
For every interaction $\op h_i$ in $\op H_0=\sum_{i=1}^N \op h_i$ as per \cref{th:main}, we add an ancilliary system $\field C^{T_i}$, where $T_i=\BigO(\poly N)$ will be specified later.
Then $\H\cl=\bigotimes_i\field C^{T_i}=:\bigotimes_i\H\cl^{(i)}$.
On each $\H\cl^{(i)}$, we define the Hamiltonian\footnote{The subscript ``clock'' stems from the standard terminology in Hamiltoinan complexity theory where the graph Laplacian part of \cref{eq:biased-clock-ham} denotes the transition terms of a so-called history state Hamiltonian.}
\begin{equation}\label{eq:biased-clock-ham}
    \op H\cl^{(i)}:= - (b_i+1)\ketbra{0} + \sum_{t=0}^{MT_i-1}(\ket t-\ket{t+1})(\bra t-\bra{t+1}),
\end{equation}
where $b_i\in[1,3]$ independent of $n$ to be specified later; this is precisely $\Hb$ from \cref{sec:H-bound}, where we emphasize the sum running form $t=0$ to $t=MT_i-1$.
As noted at the end of \cref{sec:H-bound}, $\op H\cl^{(i)}$ acts on a single qudit of dimension $MT_i$; by \cref{rem:clock-times} we can similarly define $\op H\cl^{(i)}$ to have $2$-local nearest neighbour interactions on a constant local dimension spin chain, and all of the following construction will go through unaltered.
We set $\op H\cl := \sum_{i=1}^N \op H\cl^{(i)}$.

In addition, we raise each local interaction $\op h_i$ in $\op H_0$ to couple to the $T_i$\textsuperscript{th} basis state, i.e.\ we write
\begin{equation}\label{eq:h-coupled-clock}
    \op h_i':=\op h_i\otimes(\1\otimes\ldots\otimes\1\otimes\ketbra{T_i}\otimes\1\otimes\ldots\otimes\1)=:\op h_i\otimes\ketbra{T_i}_i.
\end{equation}

We remark that $\ketbra{T_i}_i$ can be made into an at most $2$-local projector on a spin chain in a similar fashion as $\Hb$; how exactly this is done will depend on the construction used to turn $\Hb$ into a local interaction operator, and we refer the reader to \cref{rem:clock-times} and the references mentioned in the proof for more details on how this can be achieved.

The reason for choosing $\op H\cl^{(i)}$ to run to $t=MT_i-1$, and then couple $\op h_i$ to the $T_i$\textsuperscript{th} basis state is that, as per \cref{lem:any-coupling-strength}, we can very precisely control the weight $\braket{T}{\Psi_0}$ of the ground state $\ket{\Psi_0}$ of $\Hb$ if it is defined over a path graph Laplacian of length $MT$ for $M>3$.
In turn, this control will allow us to tune the effective coupling strength for the $\op h_i$ by chosing $b_i$ and $T_i$ appropriately.

\subsection{Unique Coupling Tiling}
We will use $\H\tile$ to introduce an extra coupling term to the $\op h_i'$ that will force products of two distinct terms---i.e.\ $\op h_i'\op h_j'$ for $i\neq j$---to vanish.
In principle this is straightforward; if $\H\tile$ was, say, $\field C^N$, we could introduce an orthogonal projector for each interaction via $\op h_i'\otimes\ketbra{i}$.
Then clearly $(\op h_i\otimes\ketbra{i})(\op h_j\otimes\ketbra{j})=0\ \forall i\neq j$.
The issue with this solution is that we introduced a single $N$-dimensional spin with a high interaction degree, which we want to avoid.

To circumvent this problem, we introduce an extra qutrit per interaction, i.e.\ as before $\H\tile^{(i)}:=\field C^3$.
We furthermore add one extra qutrit on the left and right side with indices $i=0$ and $i=N+1$, and set $\H\tile:=\bigotimes_{i=0}^{N+1}\H\tile^{(i)}$.
On this space, we introduce a diagonal tiling Hamiltonian \`a la
\begin{equation}\label{eq:tiling-ham}
    \op H\tile := 2\sum_{i=1}^{N-1}\big[\ketbra{21} + \ketbra{20} + \ketbra{10} + \ketbra{11}\big]_{i,i+1} - \sum_{i=0}^{N-1}\ketbra{012}_{i,i+1,i+2}.
\end{equation}
It is easy to check that all eigenvectors of $\op H\tile$ are product states of the basis $\{\ket 0, \ket 1, \ket 2\}$ (i.e.\ ternary strings), with an $N$-fold degenerate ground space
\begin{equation}\label{eq:tiling-gs}
    \mathcal L_0(\op H\tile) = \spn\{ \ket{0122\cdots2222}, \ket{0012\cdots2222}, \ldots, \ket{0000\cdots0012} \}.
\end{equation}
Observe that the states are such that there is precisely one, respectively, where a $\ket 1$ is at position $i$ for all $1\le i\le N$, and that the ground space energy is precisely $-1$, with a spectral gap of $1$.

We couple the $\op h_i'$ to $\H\tile$ with interaction terms of the form
\begin{equation}\label{eq:h-coupled-tile}
    \op h_i'':=\op h_i'\otimes(\1\otimes\ldots\otimes\1\otimes\ketbra{1}\otimes\1\otimes\ldots\otimes\1)=:\op h_i'\otimes\ketbra{1}_i,
\end{equation}
so that the overall Hamiltonian then reads
\begin{equation}\label{eq:H-approx}
    \op H' := \1\otimes\1\otimes\op H\tile + C\1\otimes\op H\cl\otimes\1 + \sum_{i=1}^N \op h_i\otimes\ketbra{T_i}_i\otimes\ketbra{1}_i,
\end{equation}
where we introduced a constant $C$ to be able to satisfy the preconditions for the Feynman-Dyson expansion: since $\op H\cl$ has a constant gap---see \cref{lem:H-bound-gs-1}---we will have to pick $C=\Omega(N)$; we will parametrize this dependence as $C=\Theta(N^{2+\delta})$, where $\delta\ge0$ is a parameter to be chosen in due course.

\subsection{Restriction to Good Signatures}\label{sec:restriction}
The first term $\1\otimes\1\otimes\op H\tile$ in \cref{eq:H-approx} commutes with all others, which means that $\op H'$ is block-diagonal with respect to the eigenstates of $\op H\tile$.
This implies that we can restrict our attention to the blocks representing the ground space of $\op H\tile$---all other blocks will have energy $\ge 1$.

\newcommand{\low}{_\mathrm{tile}}
We write $\cdot|\low$ for a restriction to the ground space $\mathcal L_0(\op H\tile)$ as defined in \cref{eq:tiling-gs}.
More specifically, we set $\op A|\low:=(\1\otimes\1\otimes\Pi\low)\op A(\1\otimes\1\otimes\Pi\low)$, where $\Pi\low$ is a projector onto $\mathcal L_0(\op H\tile)$, such that
\begin{equation}\label{eq:H-approx-simplified}
    \op H'|\low = C\1\otimes\op H\cl\otimes\Pi\low + \sum_{i=1}^N\op h_i\otimes\ketbra{T_i}_i \otimes(\ldots\ketbra 0\otimes\ketbra 1_i\otimes\ketbra 2\otimes\ldots).
\end{equation}
Observe that now products of distinct terms within the sum---those containing products $\op h_i\op h_j$ for $i\neq j$---are projected out; and further all terms from $\op H\tile$ vanished since we are within its ground space.

\subsection{Series Expansion}\label{sec:series-expansion}
As in \cite{Kempe2006}, we utilize a perturbative series expansion to estimate what the low-energy subspace of $\op H'$ looks like; for an introduction and the notation we use in the following see \cref{sec:f-d}.

By \cref{sec:restriction}, and since $\op H'$ is block-diagonal with respect to $\op H\tile$'s eigenstates, we can simplify the notation in the following analysis by only working within the subspace under the restriction $\cdot|\low$; all other eigenstates have energy $\ge 1$.
This means that we can write and partition \cref{eq:H-approx-simplified} as
\[
    \op H'|\low = \overbrace{C\sum_{i=1}^N\1\otimes\op H\cl^{(i)}}^{=:\op H} + \overbrace{\sum_{i=1}^N \op h_i\otimes\ketbra{T_i}_i}^{=:\op V},
\]
where we dropped the $\H\tile$ part of the Hilbert space; it can uniquely be reconstructed from \cref{eq:H-approx-simplified}.
We will denote the eigenstates of $\op H\cl=\sum_{i=1}^N \op H\cl^{(i)}$ with $\ket{\psi_j}$ for $j\in\{0,\ldots,\dim\H\cl-1\}$.
The ground space projector of $\op H\cl$ and its complement are then given by
\begin{alignat}{2}
    \Pi_- &= \1 \otimes \ketbra{\psi_0} =: \1 \otimes\left( \op P_{0,1}\otimes\ldots\otimes\op P_{0,N} \right) &&=: \1\otimes\op P_-   \label{eq:Pi-}\\
    \Pi_+ &= \1 \otimes \ketbra{\psi_0}^\perp &&=: \1\otimes\op P_+,   \label{eq:Pi+}
\end{alignat}
where $\op P_{0,i}$ is given by $\ketbra{\Psi_0}$ from \cref{lem:H-bound-gs-2}, for a $\op H\cl^{(i)}=\Hb$ on a chain of length $T_i$;
we further implicitly assume an energy shift to set the ground space energy of $\op H\cl$ to zero by introducing an energy shift for each individual clock Hamiltonian.

To keep the notation consistent, we will denote the eigenvectors of said $\op H\cl^{(i)}$ for a certain chain length $T_i$ with $\ket{\Psi_{j,i}}$, and the eigenvalues by $\mu_{j,i}$, for $j=0,\ldots,T_i-1$.
Then $\op P_{0,i}=\ketbra{\Psi_{0,i}}$ and $\op P_{0,i}^\perp = \sum_{j>0}\ketbra{\Psi_{j,i}}$.
We note that the $\op H\cl^{(i)}$---and hence of $\op H\cl$---are real symmetric matrices; we can therefore choose all its eigenvectors with real entries, which we will assume henceforth.

The complement projector $\op P_+=\sum_{j>0}\ketbra{\psi_j}$ is a bit more complicated to express in closed form; summing over all binary strings of length $N$ apart from the all zero string,
\begin{equation}\label{eq:P-plus}
    \op P_+=\sum_{s\neq0\cdots0}\op P_{0,1}^{(s_1)}\otimes\ldots\otimes\op P_{0,N}^{(s_N)}
    \quad\text{where}\quad
    \op P_{0,i}^{(s_i)}=\begin{cases}
    \op P_{0,i} & \text{if $s_i=0$} \\
    \op P_{0,i}^\perp & \text{otherwise.}
    \end{cases}
\end{equation}
We can re-express \cref{eq:P-plus} in terms of the eigensystems of the individual $\op H\cl^{(i)}$'s, as
\[
\op P_+ = \sum_{s\neq0\cdots0}\sum_{k_1=s_1}^{(T_1-1) s_1}\cdots \sum_{k_N=s_N}^{(T_N-1) s_N}
\ketbra{\Psi_{k_1,1}}\otimes\ldots\otimes\ketbra{\Psi_{k_N,N}},
\]
where $T_i$ is the number of eigenstates of $\op H\cl^{(i)}$, and the sums either just sum over a single term $k_i=0$ if $s_i=0$, or $k_i=0,\ldots,T_i-1$ if $s_i=1$.

The products of these projectors with some $\ketbra{T_j}_j$, $j\in\{1,\ldots,N\}$, are as follows.
\begin{align}
\op P_- &\ketbra{T_j}_j \op P_- = \bra T_j\op P_{0,j}\ket T_j \op P_- =: p_{0,j}^2 \op P_-, \label{eq:P-prods1}\\[0.5cm]
\op P_- &\ketbra{T_j}_j \op P_+ = \op P_{0,1}\otimes\ldots\otimes\op P_{0,j-1}\otimes \op P_{0,j}\ketbra{T_j}\op P_{0,j}^\perp \otimes \op P_{0,j+1}\otimes\ldots\otimes\op P_{0,N} \nonumber \\
&= p_{0,j} \op P_{0,1}\otimes\ldots\otimes\op P_{0,j-1}\otimes
\ket{\Psi_{0,j}}\sum_{i>0}\braket{T_j}{\Psi_{i,j}}\bra{\Psi_{i,j}}
\otimes \op P_{0,j+1}\otimes\ldots\otimes\op P_{0,N} \nonumber \\
&=: p_{0,j} \op P_{0,1}\otimes\ldots\otimes\op P_{0,j-1}\otimes
\ketbra{\Psi_{0,j}}{p_j}
\otimes \op P_{0,j+1}\otimes\ldots\otimes\op P_{0,N}, \label{eq:P-prods2}\\[0.5cm]
\op P_+ &\ketbra{T_j}_j \op P_+ \nonumber\\
&=\sum_{s,r\not\equiv 0}\op P_{0,1}^{(s_1)}\otimes\ldots\otimes\op P_{0,N}^{(s_N)}\left(
\1\otimes\ldots\otimes\ketbra{T_j}\otimes\ldots\otimes\1\right)
\op P_{0,1}^{(r_1)}\otimes\ldots\otimes\op P_{0,N}^{(r_N)} \nonumber\\
&= \sum_{s,r\not\equiv 0}\op P_{0,1}^{(s_1)}\delta_{s_1,r_1}\otimes\ldots\otimes\op P_{0,j}^{(s_j)}\ketbra{T_j}\op P_{0,j}^{(r_j)}\otimes\ldots\otimes\op P_{0,N}^{(s_N)}\delta_{s_N,r_n} \nonumber\\
&= \frac12\sum_{\text{all\ }s}\op P_{0,1}^{(s_1)}\otimes \ldots\otimes\Big( \nonumber\\
&\hspace{2cm}\op P_{0,j}\ketbra{T_j}\op P_{0,j}^\perp + \op P_{0,j}^\perp\ketbra{T_j}\op P_{0,j} + \op P_{0,j}^\perp\ketbra{T_j}\op P_{0,j}^\perp  \nonumber\\
&\hspace{1.5cm}\Big)\otimes\ldots\otimes\op P_{0,N}^{(s_N)} \nonumber \\
&= \frac12\sum_{\text{all\ }s}\op P_{0,1}^{(s_1)}\otimes\ldots\otimes\big( p_{0,j}\ketbra{\Psi_{0,j}}{p_j} + p_{0,j}\ketbra{p_j}{\Psi_{0,j}} + \ketbra{p_j}\big) \otimes\ldots\otimes\op P_{0,N}^{(s_N)}.\label{eq:P-prods3}
\end{align}
We emphasize that in the last two lines, we sum over all binary strings $s$, which is where the factor of $1/2$ stems from.
Again for consistency of notation, we set $p_{i,j}:=\braket{T_j}{\Psi_{i,j}}$.
Note that the $p_{i,j}$ are always real, since we chose our eigenbasis real.

We are interested in the low-energy space of $\tilde{\op H}$, for which we can calculate the expansion terms of $\Sigma_-(z)$ from \cref{eq:self-energy} using \cref{eq:Z-expansion,eq:P-prods1,eq:P-prods2,eq:P-prods3}.
We have
\begin{align}
    \op H_-
    &= \Pi_-\op H\Pi_- = 0,  \label{eq:Z-expansion-Hm}\\[0.5cm]
    \op V_-
    &= \Pi_-\op V\Pi_- = \sum_{i=1}^N
    p_{0,i}^2 \op h_i\otimes\op P_-, \label{eq:Z-expansion-Vm}\\[0.5cm]
    \op V_+
    &= \Pi_+\op V\Pi_+ \nonumber\\
    &= \frac{1}{2}\sum_{i=1}^N \op h_i\otimes
    \sum_{\text{all\ }s}\bigg( \op P_{0,1}^{(s_1)}\otimes\ldots\otimes\big[ p_{0,i}\ketbra*{\Psi_{0,i}}{p_i} \nonumber\\
    &\quad\quad +p_{0,i}\ketbra*{p_i}{\Psi_{0,i}} + \ketbra*{p_i}\big] \otimes\ldots\otimes\op P_{0,N}^{(s_N)} \bigg),  \label{eq:Z-expansion-Vp}\\[0.5cm]
    \op V_{-+}
    &=\Pi_-\op V\Pi_+ = \sum_{i=1}^Np_{0,i} \op h_i\otimes\big[
    \op P_{0,1}\otimes\ldots\otimes\ketbra*{\Psi_{0,i}}{p_i}\otimes\ldots\otimes\op P_{0,N}
    \big],  \label{eq:Z-expansion-Vmp}\\[0.5cm]
    \op G_+
    &= \Pi_+(z\1 - \op H)^{-1}\Pi_+ \nonumber\\
    &= \Pi_+\left(\1\otimes\sum_{i=1}^N (z-\lambda_i)^{-1}\ketbra{\psi_i}\right) \Pi_+ \nonumber\\
    &= \1\otimes\sum_{i>0}(z-\lambda_i)^{-1}\ketbra{\psi_i}. \label{eq:Z-expansion-Gp}
\end{align}
We note that the term $\op G_+$ is nothing but a weighted variant of the projector $\Pi_+$.
This is consistent with what we discussed in \cref{sec:f-d}: solving the self-energy $\Sigma_-(z) = \tilde{\op H}_-$ yields the low-energy part of $\tilde{\op H}$, a weighted variant of the projector $\Pi_-$.
\Cref{eq:Z-expansion-Hm,eq:Z-expansion-Vm,eq:Z-expansion-Vp,eq:Z-expansion-Vmp,eq:Z-expansion-Gp} allow us to calculate the series terms of $\Sigma_-(z)$; 
since we are still working within the ground space of $\op H\tile$ as per \cref{sec:restriction,eq:H-approx-simplified}, all cross-terms $\op h_i\op h_j$ for which $i\neq j$ are exactly zero.
Then
\begin{align*}
    \op V_{-+}\op G_+\op V_{+-} 
    & = \sum_{i=1}^N p_{0,i}^2 \op h_i^2\otimes
        \op P_{0,1}\otimes\ldots\otimes \bigg(\ketbra{\Psi_{0,i}}{p_i}   \nonumber\\
        &\quad\quad\quad \times \sum_{k>0}(z-\mu_{k,i})^{-1}\ketbra{\Psi_{k,i}}\ketbra{p_i}{\Psi_{0,i}}\bigg)\otimes\ldots\otimes\op P_{0,N}   \nonumber\\
    &= \sum_{i=1}^N p_{0,i}^2 \sum_{k>0}(z-\mu_{k,i})^{-1} \left|\braket{p_i}{\Psi_{k,i}}\right|^2 \op h_i^2 \otimes \Pi_-.
\end{align*}
Similarly
\begin{align*}
    \op V_{-+}\op G_+\op V_+\op G_+\op V_{+-}&=\sum_{i=1}^N\Bigg( p_{0,i}^2 \sum_{k>0}(z-\mu_{k,i})^{-1}  \sum_{l>0}(z-\mu_{l,i})^{-1}\Big[
     \overbrace{\braket{p_i}{\Psi_{k,i}}}^{=\braket{T_i}{\Psi_{k,i}}}  \nonumber\\
     &\quad\quad\quad\times \braket{\Psi_{k,i}}{T_i}\braket{T_i}{\Psi_{l,i}}\braket{\Psi_{l,i}}{p_i}
     \Big]\Bigg)\op h_i^3 \otimes \Pi_-\\
     &=:\sum_{i=1}^N p_{0,i}^2 \eta_i^2\op h_i^3 \otimes \Pi_-,
\end{align*}
and therefore inductively
\begin{equation}
\op V_{-+}(\op G_+\op V_+)^n\op G_+\op V_{+-} = \sum_{i=1}^N p_{0,i}^2 \eta_i^{n+1} \op h_i^{n+2} \otimes \Pi_-.
\end{equation}
The self-energy given in \cref{eq:self-energy} then reads
\begin{equation}\label{eq:self-energy-full}
\Sigma_-(z) = \sum_{i=1}^N p_{0,i}^2 \sum_{l\ge0}\eta_i^l(z)\op h_i^{l+1}\otimes\Pi_-.
\end{equation}

To finalize our proof, we will need to analyse the $z$-dependence of $\eta_i$; this is straightforward: since we shifted each individual clock Hamiltonian such that $\mu_{0,i}=0$ and with the scaling constant $C=\Omega(N^{2+\delta})$ in \cref{eq:H-approx-simplified}, we have $\mu_{k,i}>Cb_i^2/(b_i+1)\ge C/2\ \forall i>0$ by \cref{lem:H-bound-gs-1}.
For $C\ge 4$ and for all $|z|\le1$ we have $|z-\mu_{l,i}|\ge C/4\ \forall i, \forall l>0$---where the condition $C\ge 4$ simply translates into a condition on the system size $N$, which in turn depends on the proportionality constant in the Landau $C=\Omega(N^{2+\delta})$  that was free to choose in \cref{eq:H-approx}; fixing it to $C = 4 N^{2+\delta}$, for instance, yields the result for all $N\ge1$.
By \cref{cor:H-bound-T-overlap} we then get
\begin{equation}\label{eq:C-bound}
    |\eta_i| \le \sum_{l>0}|z-\mu_{l,i}|^{-1} \left|\braket{T_i}{\Psi_{l,i}}\right|^2 
    \le \frac{4}{C}\left[1+\BigO\left(\frac{b_i(b_i+2)}{(b_i+1)^{2T_i}}\right)\right]
    =\BigO\left(C^{-1}\right).
\end{equation}
Note that we arbitrarily chose the region of $z$ to have radius $1$; this has to do with our choice of $b_i\in[1,3]$, which itself is arbitrary; tuning the norm of some $\op H\eff$ will then have to be done by making $T_i$ larger, see \cref{lem:any-coupling-strength}.

\subsection{Proof of Main Result}\label{sec:proof}
\begin{proof}[\Cref{th:main}]
In order to proof \cref{th:main}, we start with a $k$-local Hamiltonian $\op H_0 = \sum_{i=1}^Nr_i\op h_i$ on a Hilbert space $(\field C^d)^{\otimes n}$, where each $\| \op h_i \| = 1$ and $| r_i(n) |\le r(n)\ \forall i$.
We assume without loss of generality that the $\op h_i$ square to identity,\footnote{%
A canonical basis for the Hermitian $d\times d$ matrices is given by the $d$ linearly inependent matrices $\{\op e_i \}_{i\in\{1,\ldots,d\}}$ such that $\op e_i$ has a single $1$ on the diagonal at the $i$\textsuperscript{th} location, as well as the $d(d-1)/2$ matrices $\{\op e_{i,j} \}_{1\le i<j\le d}$ and $\{ \op e'_{i,j} \}_{1\le i<j\le d}$, where $\op e_{i,j}$ has a matching pair of $1$s at location $(i,j)$ and $(j,i)$, and similarly $\op e'_{i,j}$ a $(\mathrm{i},-\mathrm{i})$-pair on corresponding off-diagonal locations $(i,j)$ and $(j,i)$, respectively.
We can define a new set of operators as follows:
\begin{align*}
\op f_i :=&\ \1_d - 2\op e_i \quad\text{for}\ i=1,\ldots,d \\
\op f_{i,j} :=&\ \1_d + \op e_{i,j}/\sqrt 2 -(1 + 1/\sqrt 2)\op e_i - (1 -1/\sqrt 2)\op e_j  \quad\text{for}\ 1\le i<j\le d\\
\op f'_{i,j} :=&\ \1_d + \op e'_{i,j}/\sqrt 2 -(1 + 1/\sqrt 2)\op e'_i - (1 - 1/\sqrt 2)\op e'_j \quad\text{for}\ 1\le i<j\le d.
\end{align*}
It is easy to verify that these operators are all hermitian, form a basis of the $d\times d$ hermitian matrices, and all square to $\1_d$.
} i.e. we demand $\op h_i^2 = \1$ for all $i\in\{1,\ldots,N\}$.
We set $r'_i(n) := r_i(n) / (200 r(n))$, each of which now satisfies $r_i \in (0,1/100)$.
With the local terms $\op h_i$, we define $\op V$ for a new $k+2$-local Hamiltonian $\op H'$ as in \cref{eq:H-approx}, where $\op H\cl=\sum_{i=1}^N \op H\cl^{(i)}$;
by \cref{lem:any-coupling-strength}, we know that, for all $i\in\{1, \ldots, N\}$, there exist parameters $b_i$, $M$, and $T_i$ for $\op H\cl^{(i)}$ such that 
\[
p_{0,i}^2  = \left|\braket{\Psi_{0,i}}{T_i}\right|^2 = \frac{r_i}{200 r(n)} \in (0, 1/100).
\]

Set $\op H\eff := \op V_- $.
By \cref{eq:Z-expansion-Vm}, we then have
\[
\op H\eff = \sum_{i=1}^N p_{0,i}^2 \op h_i \otimes \Pi_- = \left(\sum_{i=1}^N \left| \braket{\Psi_{0,i}}{T_i} \right|^2 \op h_i \right) \otimes \Pi_-  =  \frac{1}{200 r(n)} \op H_0 \otimes \Pi_-,
\]
with $\Pi_-$ defined in \cref{eq:Pi-}.
Furthermore, \cref{eq:self-energy-full,eq:C-bound} tell us that
\begin{align}
    \epsilon := \| \op H\eff - \Sigma_-(z) \|
    &\overset{*}{=} \left\| \sum_{i=1}^N p_{0,i}^2 \sum_{l\ge 2} \eta_i^l(z) \op h_i^{l+1} \otimes \Pi_- \right\|    \nonumber\\
    &\le \frac{N}{100} \sum_{l\ge2} \BigO\left( C^{-l} \right)    \nonumber\\
    &=\BigO(N) \sum_{l\ge2} N^{-(2+\delta)l} =  \frac{\BigO(N)}{N^{4+2\delta}-N^{2+\delta}} = \BigO\left(N^{-3-\delta} \right).   \label{eq:epsilon}
\end{align}
where in the first line $(*)$ we used the fact that the term of order $l=1$ in the second sum just introduces a constant energy shift---as by assumption $\op h_i^2=\1\ \forall i$.
The Landau $\BigO$ terms are with respect to the limit $N\longrightarrow\infty$.

Let us now remove the tiling Hamiltonian from \cref{sec:restriction} and reduce the extra locality introduced in \cref{eq:H-approx} by $1$; we call this Hamiltonian $\tilde{\op H}$.
More explicitly, we now lift the implicit assumption of working in the ground space of $\op H\tile$, within which all cross terms $\op h_i\op h_j$ vanish for $i\neq j$.
This means that at order $l$ in the above sum defining $\epsilon$, we will get at most $N^l$ additional cross-terms to take care of, all of which of unit norm within the sum in \cref{eq:self-energy-full}.
A short calculation yields the final error bound $\epsilon' = \BigO(N^{-2-\delta})$ for $\tilde{\op H}$.

Invoking \cref{th:pert-2}, we get
\begin{align*}
    \| \tilde{\op H}_- - \op H\eff \| &\le \BigO\left(
        \frac{N(\|\op H\eff\| + \epsilon')}{\lambda_+}
    \right) + \BigO(\epsilon') \\
    &=\BigO\left(
            \frac{N^{2}/r(n) + N^{-1-2\delta}}{N^{2+\delta}}
        \right) + \BigO(\epsilon') \\
    &=\BigO\left(N^{-\delta}\right),
\end{align*}
where we used $\| \op H\eff \| = \BigO(N/r(n))$, $\| \op V \| = \BigO(N)$, and $\lambda_+$ as the spectral gap of $\op H\cl$---which scales as $C$.

What is left to show now is that the local dimension of the ancilliary system necessary to specify $\op H\cl$ is as claimed for the two cases of scaling of $r(n)$---i.e.\ $q=3$ if $r=\BigO(\exp(\poly n))$, and $q=9$ for $r=\BigO(\exp(\exp(\poly n)))$.
This follows by \cref{rem:clock-times}, which concludes the proof.
\end{proof}

\section{Applications, Extensions and Corollaries}

\subsection{The Local Hamiltonian Problem}
Hamiltonian complexity theory has spawned a whole host of literature and research, from hardness proofs
\cite{Oliveira2008,Aharonov2009,Bausch2016,Bausch2017,Bravyi2006,Bravyi2006a,Schuch2011,Kempe2006,Hallgren2013,Gharibian2015},
efficient algorithms \cite{DeBeaudrap2016,Arad2015,Arad2013,Landau2013}, modified proposals on encoding computation into the ground state of a local Hamiltonian \cite{Breuckmann2013,Caha2017,Bausch2016a,Usher2017}, to suggestions on how to perform quantum computation with a Hamiltonian \cite{Wei2015,Nagaj2008,Nagaj2012}, or simulation and universality \cite{Cubitt2017,Piddock2015,Cubitt2013,Childs2010}, just to name a few.
In order to satisfy the task for physically realistic models---typically translational invariance and low local dimension---it is often necessary in these constructions to break down many-body terms into two-body terms.
The traditional method is to use perturbation gadgets, which, as discussed extensively, introduces energy scales that scale both in the required absolute error, as well as in the interaction range.

Can we apply our methods to improve upon one of the existing results?
In the following subsections we will pick a representative problem of each class and discuss the respective implications.

The \lham problem is the complexity-theoretic formalization of the question of approximating the ground state energy of a local Hamiltonian \cite{Kitaev2002}, which is a natural question that arises in physics.
It is the quantum analogue of classical boolean satisfiability problems such as \sat: while the latter asks for an assignment to boolean variables that render a logic statement true, \lham asks how well a quantum state can satisfy local constraints (given by the local interaction terms of some local Hamiltonian $\op H=\sum_{i=1}^N \op h_i$).
Kitaev proved that this problem is complete for the complexity class \QMA, by a construction first introduced by Feynman \cite{Feynman1986}.
Completeness for \QMA implies that on a quantum computer one can \emph{verify} a solution efficiently within poly-time and with success probability $\ge2/3$.
Just like \NP, \QMA makes no claims about \emph{obtaining} said solution in first place.

To be precise about all the parameters involved, we give the formal definition of \lham, as well as the complexity classes \QMA, \QMAEXP, and \BQEXPSPACE, for which we will prove hardness results of variants of the \lham problem in the following; for a brief but detailed reference of complexity-theoretic terminology, as well as the notion of Turing machines and quantum cirucuits, we refer the reader to \cite{Watrous2012}, in which the following definitions can also be found.
\begin{definition}[\lham]\label{def:lham}\leavevmode\\
\textbf{Input:} $k$-local Hamiltonian $\op H=\sum_{k=1}^N \op h_i$ on $(\field C^d)^{\otimes n}$, $N=\poly n$, $\| \op h_i \|=\poly n\ \forall i$. Two real numbers $\alpha,\beta$ with $\beta-\alpha\ge1/\poly n$.\\
\textbf{Promise:} The ground state energy of $\op H$ satisfies either $\lmin(\op H)\ge\beta$, or $\lmin(\op H)\le\alpha$.
\textbf{Output:} \yes iff $\lmin(\op H)\le\alpha$.
\end{definition}
Note that while \cref{def:lham} does not allow local terms to have exponentially large norm, it \emph{does} allow exponentially small norms; yet not more as the bit complexity of the matrix entries---which comprise the input to the \lham problem---have to be bounded by a polynomial.

\begin{definition}[Promise Problem]
Let $\Sigma$ be a finite set, called alphabet.
A promise problem is a set $A \subseteq \Sigma^*$---where the ${}^*$ denotes the Kleene star, i.e.\ strings of symbols of $\Sigma$ of length $\ge 0$---such that $A=A_\yes \dot\cup A_\no$, called \yes- and \no-instances, respectively.
\end{definition}

In the following, we will always assume that $\Sigma = \{0, 1\}$, and we identify $\ket x := \ket{x_0x_1\cdots x_{n-1}} \in (\field C^2)^{\otimes n}$ for some instance $x\in A$, $|x|=n$.
\begin{definition}[\BQP and \BQEXP]\label{def:bqp}
If there exists a polynomial-time terminating Turing machine which for all $n\in\field N$, on input $1^n$, writes out the description of a quantum circuit $Q_n$, we call the family $Q=\{ Q_n\}_{n\in\field N}$ polynomial-time generated, or polynomial-time uniform.
A promise problem $A$ is in $\BQP(a,b)$ for functions $a,b:\field N \longrightarrow [0, 1]$ if there exists a polynomial-time uniform quantum circuit family $Q$, such that $Q_n$ acts on an $n$ qubit input $\ket x$, $x\in A$ with $|x|=n$ and has a single measured output qubit measured either in state $\ket 0$ or $\ket 1$, where the latter signifies ``accept'', which we write $Q_n(\ket x)=1$.
The circuit family satisfies
\begin{enumerate}
\item $\Pr(Q_n(\ket x) = 1) \ge a(n)$ if $x\in A_\yes$, or otherwise
\item $\Pr(Q_n(\ket x) = 1) \le b(n)$ if $x\in A_\no$.
\end{enumerate}
By convention $\BQP = \BQP(2/3, 1/3)$.
\BQEXP is defined analogously, replacing polynomial time with exponential time (strictly speaking $\BigO(\exp(n^c))$-time, for any constant $c\ge 0$) throughout.
\end{definition}
\begin{definition}[\QMA and \QMAEXP]\label{def:qma}
A promise problem $A$ is in $\QMA_p(a,b)$ if for the same setup as in \cref{def:bqp}, $Q_n$ acts on an input of size $n+p(n)$ for some $p(n) = \poly n$ and a single output qubit, such that
\begin{enumerate}
\item $\forall x\in A_\yes\ \exists \ket{\psi} \in (\field C^2)^{\otimes p(n)} : \Pr(Q_n(\ket x, \ket \psi) = 1) \ge a(n)$, and
\item $\forall x\in A_\no\ \forall \ket{\psi} \in (\field C^2)^{\otimes p(n)} : \Pr(Q_n(\ket x, \ket \psi) = 1) \le b(n)$.
\end{enumerate}
We set $\QMA = \bigcup_{p(n) = \poly n} \QMA_p(2/3,1/3)$.
The circuit family $Q_n$ is also called \emph{verifier} (which itself is a \BQP circuit with an extra unconstrained input), and the quantum state $\ket\psi$ a witness for the instance; as in \cref{def:bqp}, we define \QMAEXP in a similar fashion, replacing the \BQP verifier with a \BQEXP one.
\end{definition}
We note that one can amplify the acceptance and rejection probabilities of $2/3$ and $1/3$ in \cref{def:bqp,def:qma} such that $\BQP=\BQP(1-2^{-q}, 2^{-q})$, for any $q(n) = \poly n$ in the input size $n$ \cite[Prop.~3]{Watrous2012}.
\StoqMA is defined as \QMA, but for a classical reversible boolean verifier circuit instead of a quantum circuit (i.e.\ a \P verifier), inputs restricted to $\ket{0}$ and $\ket+ = (\ket0 + \ket1)/\sqrt 2$ and a final measurement in the basis $\{\ket+,\ket-\}$---where $\ket- := (\ket0-\ket1)/\sqrt 2$.
If we remove randomness completely we end up with the complexity classes \P and \NP, of which \BQP and \QMA are the natural quantum analogues.

Instead of bounding the computational runtime, one can in a similar fashion bound the required space; yet instead of uniform families of quantum circuits a hybrid model of a classical Turing machine which can perform quantum operations on a separate tape of qubits is a more natural notion; the space requirements for such a quantum Turing machine is defined by how much classical and quantum tape the machine ingests during a computation; we again refer the reader to \cite[Sec.~VII.2]{Watrous2012} for an extended introduction.
\begin{definition}[\BQPSPACE and \BQEXPSPACE]\label{def:BQEXPSPACE}
A promise problem $A$ is in \BQPSPACE if there exists a quantum Turing machine with poly-bounded space requirement, accepting \yes instances with probability $\ge 2/3$, and \no instances with probability $\le 1/3$.
\BQEXPSPACE is defined analogously.
\end{definition}
What might come as a surprise is that, in contrast to the amplification statement for \BQP---which limits how close to $1/2$ acceptance and rejection probabilities may lie---\BQPSPACE$\!=$\PQPSPACE, defined with $>1/2$ and $\le 1/2$ acceptance and rejection probabilities.
Even more surprisingly, \BQPSPACE$\!=$\PSPACE \cite{Watrous2003}---i.e.\ classical computers (without access to randomness) are as powerful as quantum computers, given the only restriction is placed on how much space each machine is allowed to demand.

The \lham problem as defined in \cref{def:lham} is known to be \QMA-complete \cite{Kitaev2002}; and as mentioned, variants of this result have been proven which impose ever more restrictions onto the types of Hamiltonians for which the same result holds.
For instance, for a promise gap (i.e.\ the difference $\beta-\alpha$ in \cref{def:lham}) which closes as $\propto 1/\exp n$, the \lham problem is known to be \PSPACE-complete \cite{Fefferman2016}; this is shown by encoding a variant of a \QMA-hard problem with an acceptance and rejection probability $>1/2$ and $\le 1/2$, respectively, matching the probabilistic bounds in the definition of \BQPSPACE.
Another variant is for the case of translationally-invariant local Hamiltonians for which the \lham problem is \QMAEXP-complete \cite{Gottesman2009}: this is due to the fact that the specification of an instance has bit complexity $|n|$ in the system's size $n\in\field N$---since this is the only free variable in a translationally-invariant system that changes from instance to instance.
To still obtain a polynomially-closing promise gap we need to allow the verifier circuit to run for an exponential time (cf.\ \cite[Sec.~3.4]{Bausch2016}).

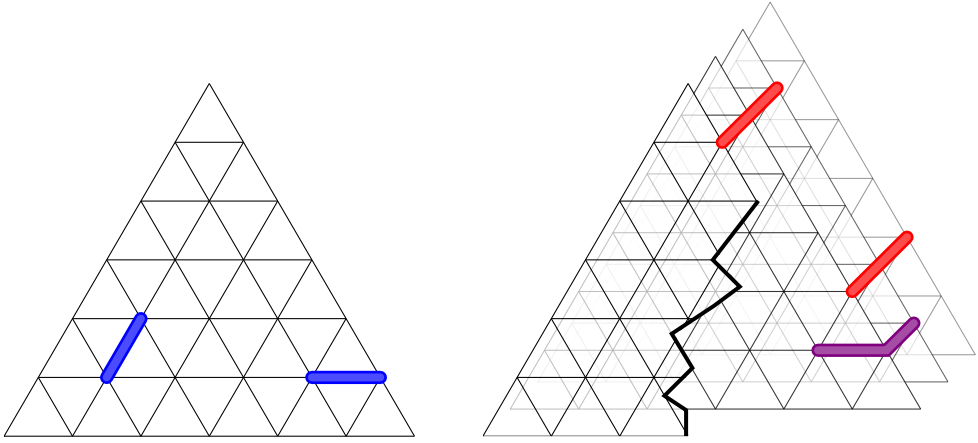
\begin{figure}
	\centering
\newcommand*\rowss{6}
\newcommand{\trilattice}[2]{%
    \fill[white,opacity=.7] (0,0,#1) -- (\rowss,0,#1) -- (\rowss/2,\rowss*.866,#1) -- cycle;
    \foreach \r in {0,1,...,\rowss} {
        \draw[opacity=#2] ($(0,0,#1)+\r*(0.5, 0.866,0)$) -- ($(\rowss,0,#1)+\r*(-0.5, 0.866, 0)$);
        \draw[opacity=#2] ($(0,0,#1)+\r*(1, 0, 0)$) -- ($(\rowss/2,\rowss*.866,#1)+\r*(0.5,-.866, 0)$);
        \draw[opacity=#2] ($(0,0,#1)+\r*(1, 0, 0)$) -- ($(0,0,#1)+\r*(0.5,.866, 0)$);
    };
}
\begin{tikzpicture}[x=1cm, y=1cm, z=.4cm,scale=.9]
	\clip (0,0) rectangle (14.5,6.5);
    \trilattice{0}{1}
    \draw[line width=5pt,blue,line cap=round] (\rowss-1.5,.866,0) -- (\rowss-.5,.866,0);
    \draw[line width=3pt,white,line cap=round,opacity=.3] (\rowss-1.5,.866,0) -- (\rowss-.5,.866,0);
    \draw[line width=5pt,blue,line cap=round] (1.5,.866,0) -- (2,2*.866,0);
    \draw[line width=3pt,white,line cap=round,opacity=.3] (1.5,.866,0) -- (2,2*.866,0);
    \begin{scope}[xshift=7cm]
        \foreach \x/\a in {3/.4,2/.6,1/.8} {
            \trilattice{\x}{\a}
        }
     \begin{scope}
	        \clip (0,0) -- (3,0) -- (3,.4) -- (2.7,.6) -- (3.1,1) -- (2.8,1.5) -- (3.4,1.9) -- (3.8,2.2) -- (3.4,2.6) -- (4.05,3.45) -- (4.5,3.7) -- (10,10) -- (0,10);
	        \trilattice{0}{1}
	        \draw[black,line width=3pt] (3,0) -- (3,.4) -- (2.7,.6) -- (3.1,1) -- (2.8,1.5) -- (3.4,1.9) -- (3.8,2.2) -- (3.4,2.6) -- (4.05,3.45);
        \end{scope}
        \draw[line width=5pt,red!50!blue,line cap=round] (\rowss-1.5,.866,1) -- (\rowss-.5,.866,1) -- (\rowss-.5,.866,2);
        \draw[line width=3pt,white,line cap=round,opacity=.3] (\rowss-1.5,.866,1) -- (\rowss-.5,.866,1) -- (\rowss-.5,.866,2);
        \draw[line width=5pt,red,line cap=round] (\rowss-1,2*.866,1) -- (\rowss-1,2*.866,2) -- (\rowss-1,2*.866,3);
        \draw[line width=3pt,white,line cap=round,opacity=.3] (\rowss-1,2*.866,1) -- (\rowss-1,2*.866,2) -- (\rowss-1,2*.866,3);
        \draw[line width=5pt,red,line cap=round] (\rowss-2.5,5*.866,0) -- (\rowss-2.5,5*.866,1) -- (\rowss-2.5,5*.866,2);
        \draw[line width=3pt,white,line cap=round,opacity=.3] (\rowss-2.5,5*.866,0) -- (\rowss-2.5,5*.866,1) -- (\rowss-2.5,5*.866,2);
    \end{scope}
\end{tikzpicture}
\caption{Triangular lattice, and stacked triangular lattice, used in \cref{th:piddock,th:cor-1}, respectively.
The blue line indicates a $2$-local interaction between spins in the same triangular lattice layer; the purple line a $3$-local interaction emerging from the extra coupling between two lattice layers.
The red $3$-local interaction represents the highest locality terms within $\op H\cl^{(i)}$.}
\label{fig:lattice}
\end{figure}

Returning from this digression, we now wish to analyse whether we can improve upon any of these best-known results in some aspect.
To this end, we will focus on a concrete example, namely \citeauthor{Piddock2015}'s proof that the \lham problem is \QMA-complete, even with antiferromagnetic interactions on a triangular lattice \cite{Piddock2015}.
\begin{theorem}[\citeauthor{Piddock2015} \mbox{\cite[Th.~4]{Piddock2015}}]\label{th:piddock}
Let $(V,E)$ be a triangular lattice of $|V|=n$ vertices, as shown in \cref{fig:lattice}.
Let $\alpha,\beta$ and $\gamma$ such that $\alpha+\beta,\beta+\gamma, \text{and\ }\gamma+\alpha\ge 0$, and not $\alpha=\beta=\gamma$.
Then there exists a family of real positive numbers $\{ r_e \}_{e\in E}$, $r_e = \poly |V|$, such that the \lham problem for the family of Hamiltonians (indexed by the lattice's size $n$)
\begin{equation}\label{eq:universal-h}
\op H:= \sum_{e \in E}\op h_e
\quad\text{where}\quad
\op h_e:=r_e(\alpha \sigma^x\sigma^x + \beta \sigma^y\sigma^y + \gamma \sigma^z\sigma^z),
\end{equation}
is \QMA-complete.
\end{theorem}
Our goal is to employ \cref{th:main} to remove the explicit variation in coupling strength in \cref{th:piddock} given by the $r_e=\poly n$ at every lattice edge for a triangular lattice on $n$ vertices, and prove a variant of the result with a scaling limited to $\propto n^{2+\delta}$, for an arbitrarily small $\delta>0$.

\begin{theorem}\label{th:cor-1}
    Let $\Lambda=(V,E)$ be a triangular lattice as shown in with $|V|=n$ vertices, as shown in \cref{fig:lattice}, and let $\delta'>0$.
    Then for $n'\in\field N$, stacks of the lattice are given by $\Lambda' = \Lambda\square\Lambda_2$, where $\Lambda_2$ is a path graph of length $n'$, and $\square$ denoting the Cartesian graph product.
    The \lham problem is \QMA-complete with interactions on a graph $\Lambda'$, even when restricted to the following type of interactions:
    \begin{enumerate}
    \item $3$-local interactions of the form $\op h\otimes\ketbra 0$, where $\op h$ is given in \cref{eq:universal-h} but such that $\| \op h \| = 1$; $\op h$ only acts within a lattice layer $\Lambda$, and $\ketbra 0$ is a one-local projector onto state $\ket0$ of an adjacent qubit in the next higher layer;
    \item $\op q$ are diagonal geometrically $3$-local terms from \cref{eq:biased-clock-ham}, acting on the vertical edges within $\Lambda'$, such that $\| \op q \| = \BigO(s^{\delta'}))$, where $s=nN'$ is the number of vertices in $\Lambda'$.
    \end{enumerate}
\end{theorem}
\begin{proof}
    Let $\H:=(\field C^2)^{\otimes \Lambda}$ and similarly $\op H'$ be the associated Hilbert space for qubits located at each of $\Lambda$ and $\Lambda'$'s vertices, respectively.
    We start with a \QMA-complete $2$-local Hamiltonian $\op H_0 = \sum_{e \in E} \op h_e$  on $\H$ given by \cite[Th.~4]{Piddock2015}; then by construction all interactions on the triangular edges $\op h_e$ satisfy \cref{eq:universal-h}, and such that
    \begin{equation}\label{eq:r}
        \max_{i,j\in E}\{\| \op h_i \| / \| \op h_j \|\} = \poly n
        \quad\text{and}\quad
        r(n) := \max_{i\in E} \{ \| \op h_i \| \} = \poly n.
    \end{equation}
    
    By \cref{th:main}, we thus know that there exists a $3$-local Hamiltonian $\op H'$ on $\H\otimes\H_2$ with the following properties:
    \begin{enumerate}
    \item $\H_2=(\field C^3)^{N'}$, $N' = \poly n$.
    \item $\op H'$ approximates $\op H_0$ within its low-energy subspace, according to definition \cref{def:approximate}, to relative precision $\BigO(N^{-\delta})$; this means
    \[
        \Pi_- \op H' \Pi_- = \op H_0 \ketbra{\Psi_0} + \BigO\left( \frac{r(n)}{N^{\delta}}\right),
    \]
    where $\Pi_-$ are projectors onto the lower part of the spectrum of $\op H'$, for some state $\ket{\Psi_0}$ defined on an ancilliary space $\H_2$, and $\epsilon = N^{-\delta}$.
    \item $\op H'=\sum_{i=1}^{N'} \op q_i$ is $2$-local, where $1 \le \| q_i \| \le n^{2+\delta}$.
    \end{enumerate}
    To determine $\delta$, we note that by \cref{def:lham} there is a promise gap $p(n) := \beta(n) - \alpha(n) = 1/\poly n$ associated to $\op H_0$.
    In order to retain QMA-hardness of $\op H'$, we need to choose $\delta = \delta(n)$ such that $r(n)/N^\delta < p(n)$; we will therefore increase $N'$ (i.e.\ the number of triangular lattice stacks) by an at most polynomial factor---uncoupled to the rest of the system---to ensure $\| \op q \| = \BigO(n N')$.
    
    What is left to show is that the $\op H\cl^{(i)}$ Hamiltonians can be chosen such that they feature $3$-local qubit interactions, instead of $2$-local qutrit ones.
    This is straightforward: since the maximum norm ratios we need to approximate are $r(n)=\poly n$, and the overlap in \cref{lem:any-coupling-strength} is exponentially small in $T$, it suffices to have $T = \BigO(\log \poly n)$.
    To construct $\Hb$ in \cref{eq:H-bound} with $3$-local interactions on $\H_2^{(i)} = (\field C^2)^{\otimes MT}$ (for some constant $M$ as explained at the end of \cref{sec:H-bound}), we can identify
    \[
        \ket t = \ket*{\underbrace{11\ldots1}_{t\ \text{times}}0\ldots 00}
    \]
    where $\{ \ket0, \ket1 \}$ are a basis for $\field C^2$; the identification implies that the terms $\ketbra{t}{t+1}$ in $\Hb$ are three-local at most, as is easily verified; similarly, the bonus term $\ketbra 1$ can be identified with a $1$-local term $\ketbra 0$ acting on the second qubit on $\H_2$.
    
    Since every $\op H\cl^{(i)}$ has an individual $T_i$---tuned to yield an amplitude $|\braket{\Psi_{0,i}}{T_i}|^2 \propto r_i$---we need to offset the $3$-local terms in $\op H\cl^{(i)}$ such that $\ketbra{T_i}$ aligns with the triangular layer $\Lambda$ on which $\op H_0$ is defined.
    As we are free to choose said layer---as $\op H'$ does not have to be translationally invariant in this construction---the claim of the theorem follows.
\end{proof}

We remark that instead of varying the offset of $\op H\cl^{(i)}$ for each interaction $\op h_i$ in $\op H_0$ individually, we can align them all uniformly with a fixed $T_i=T$ for all $i\in\{1,\ldots,n\}$.
To see this, note that by \cref{cor:H-bound-gs}, the coupling strength induced by $\Hb$ goes asymptotically like $\propto b(b+2)/(b+1)^{2T}$.
Any pair of biases $b$, $b'$ for fixed $T$ thus allows a ratio of
\begin{equation}\label{eq:vary-only-b}
    R(b,b')=\frac{b(b+2)}{b'(b'+2)} \left(\frac{b'+1}{b+1}\right)^{2T}.
\end{equation}
We have $R(1,1)=1$, and $R(b,1)$ scales exponentially in $T$, so the claim follows as in \cref{lem:any-coupling-strength}, where we note that the overall effective Hamiltonian will be rescaled by only a polynomial factor, keeping the conditions on the promise gap in \cref{def:lham} satisfied.

As a short digression for the familiar reader, we emphasize that this result is weaker than it seems: \QMA-hardness constructions, which are based on embedding a \QMA-verifier computation into the ground state of a local Hamiltonian, are commonly given with a promise gap that scales as $\propto 1/\tau^2$ in the runtime $\tau$ of this embedded computation (see \cite{Bausch2016a}; we further point out the connection to our bound state Hamiltonian in  \cref{sec:H-bound}).
For \QMA, the runtime is thus $\tau=\poly n$ for a system size $n$.
In order to lift the promise gap arbitrarily close to constant in the system size, it always suffices to add a polynomially-sized non-interacting ancilliary space of size $n'=\poly n$; if we express $\tau$ in $n'$, we can thus get a runtime scaling $\tau=n'^{1/a}$, for some arbitrarily large $a>0$, and the promise gap thus similarly follows $\Omega(n'^{-2/a})$.

In essence, this is an artefact of Karp-reductions allowing a polynomial overhead---which work either way, i.e.\ one can shrink the input to a problem by a polynomial, reducing the runtime of a \QMA-hard construction in whatever parameter one chose to express the input size with, while maintaining the complexity-theoretic implications.
However, while the promise gap can be made to close like the $\Omega(n'^{-2/a})$ for arbitrarily large but constant $a$, constant relative promise gap (relative in the system size) would imply a quantum analogue of the classical PCP theorem.\footnote{PCP stands for ``probabilistically checkable proof'', and the PCP theorem sais that any \NP-hard problem can be verified to arbitrary precision with only constant query complexity, independent of the problem size. As explained in the introduction, the local Hamiltonian problem is the quantum analogue of \sat; a constant relative promise gap would thus imply a similar argument about a constant number of constraint violations sufficing to verify a QMA-hard problem.}

Yet instead of the necessity feature multiple, potentially wildly varying coupling strengths, \cref{th:cor-1} shows that it suffices to have a \emph{single} additional energy scale $\propto n'^{1/a}$, instead of multiple ones; all other interactions are $\BigO(1)$, independent of the system size.

\Cref{th:cor-1,eq:vary-only-b} are interesting for another reason.
The reader might have noticed by now that our construction allows us to amplify a constant-range $b\in[1,3]$ to an energy scale that varies like $b^{f(n)}$, for $f$ being a polynomial or exponential in the system size $n$.
So what if we turn this problem around, and drastically limit the range for the biases $b$, say, to an interval $b\in(1,1+\chi)$, for $\chi$ very small?
We will address this question in the next section.

\subsection{Noise Amplification and Translational Invariance}
As outlined at the end of the previous section, we want to restrict the biases present in $\op H\cl^{(i)}$ to satisfy $b_i\in(1,1+\chi)$, for $\chi^{-1}\gg1$ and all $i\in\{1,\ldots,N\}$.
What range of effective coupling strengths for a target Hamiltonian $\op H_0 = \sum_{i=1}^N \op h_i$ can emerge from these subtly-varying one-local terms inside $\op H\cl^{(i)}$?
We collect this result in the next lemma.

\begin{table}
    \begin{tabular}{r C{3cm}|C{1.05cm}|C{1.05cm}|C{3cm}}
    	\toprule
    	   $T=\Theta(\cdot)$ & \multicolumn{2}{c|}{$n^a$} & \multicolumn{2}{c}{$a^n$} \\\midrule
    	$\chi=\Theta(\cdot)$ & $n^{-b}$ &       $b^{-n}$       & $n^{-b}$ &       $b^{-n}$       \\ \midrule
    	   $R(1+\chi,1)=\Omega(\cdot)$ & $\begin{cases}
    	   	2^{-T\chi} & \mathrm{if\ } a>b  \\
    	   	1          & \mathrm{otherwise}
    	   \end{cases} $ & $1$ & $2^{-T\chi}$      & 
    	   $\begin{cases}
    	       	   	2^{-T\chi} & \mathrm{if\ } a>b  \\
    	       	   	1          & \mathrm{otherwise}
    	       	   \end{cases}$ \\ \bottomrule
    \end{tabular}
    \caption{Overview over asymptotic scaling of the achievable effective coupling ratio $R(b,1)$ for $b=1+\chi$ defined in \cref{eq:vary-only-b}, as proven in \cref{lem:noise-amplification}.
    The first two rows show the asymptotic behaviour of $T$ and $\chi$ in the system size $n$---either power-law or exponential for $a,b\ge1$; the last row shows the resulting asymptotic scaling of $R(1+\chi,1)$.
    All Landau symbols are taken with respect to the limit $n\longrightarrow\infty$.}
    \label{tab:noise-amplification}
\end{table}

\begin{lemma}\label{lem:noise-amplification}
    Let the setup be as in \cref{cor:H-bound-gs-ex}, with $\Hb$ defined as in \cref{eq:H-bound}.
    Let $\chi:\field N \longrightarrow (1,\infty)$ and $T:\field N \longrightarrow\field N$.
    Denote with $R(b,1)$ the relative achievable scaling ratio for some bias $b\ge 1$ as defined in \cref{eq:vary-only-b}.
    Then the asymptotic ratio with respect to $n\longrightarrow\infty$ is given by $R(1+\chi,1)=\BigO(f(T,\chi))$ where
    \[
        f(T,\chi) = \begin{cases}
        2^{-T\chi} & \raisebox{-\normalbaselineskip+.5\jot}{\shortstack[l]{%
        $T=\Theta(n^a) \land \chi=\Theta(n^{-b}) \land a>b\ge1$ or\\
        $T=\Theta(a^n) \land \chi=\Theta(b^{-n}) \land a>b\ge1$ or\\
        $T=\Theta(a^n) \land \chi=\Theta(n^{-b})$,
        }} \\[.7cm]
        1 & \raisebox{-\normalbaselineskip+.5\jot}{\shortstack[l]{%
                $T=\Theta(n^a) \land \chi=\Theta(n^{-b}) \land b\ge a\ge1$ or\\
                $T=\Theta(a^n) \land \chi=\Theta(b^{-n}) \land b\ge a\ge1$ or\\
                $T=\Theta(n^a) \land \chi=\Theta(b^{-n})$.
                }}
        \end{cases}
    \]
\end{lemma}
\begin{proof}
We first note
\[
    R(1+\chi,1) = \frac{b(b+2)}{3}\left(\frac{2}{(b+1)}\right)^{2T} = \frac{(1+\chi)(3+\chi)}{3} \left(\frac{2}{2+\chi} \right)^{2T}.
\]
If both $T$ and $\chi$ are power-laws, i.e.\ $T=\Theta(n^a)$, $\chi=\Theta(n^{-b})$ for $a>b>1$, then an explicit calculation shows $R(1+\chi,1) = \BigO(2^{-n^{a-b}}) = \BigO(2^{-T\chi})$.
The other cases follow in a similar fashion.
\end{proof}
An overview over the asymptotic scalings in \cref{lem:noise-amplification} can be found in \cref{tab:noise-amplification}.
One immediate corollary is the following.
\begin{corollary}\label{cor:amplification}
Take any \QMA or \QMAEXP-hard \lham problem $\op H_0 = \sum_{i=1}^N\op h_i$ (e.g.\ the construction used to prove \cref{th:cor-1}) on an $n$-partite Hilbert space $\H=(\field{C^d})^{\otimes n}$ with $N=\poly n$ local terms,
and such that $\max_{i,j\in \{1,\ldots,N\}}\{ \| \op h_i \| / \| \op h_j \| \} = \BigO(\exp(\poly n))$.
Let $\delta>0$.
Then for any $\chi=1/\poly n$,
there exists a \lham variant $\op H' = \sum_{i=1}^{N'} \op q_i$ on an $n'$-partite Hilbert space $\H' = (\field C^{d'})^{\otimes n'}$ with $n'=\poly n$, such that
\begin{enumerate}
\item each local term $\op q_i$ has norm $\| \op q_i \| \in \{1\} \cup [n'^{\delta},(1+\chi)n'^{\delta}]$,
\item the variant has a promise gap $1/\poly n'$,
\item it is \QMA (\QMAEXP) hard if $\op H_0$ is \QMA (\QMAEXP) hard, and
\item if the original variant was $2$-local, then $d'=\max\{ 3, d\}$; otherwise $d'=d$.
\end{enumerate}
\end{corollary}
\begin{proof}
We apply \cref{th:main}, but restrict the $b_i$ in $\op H\cl^{(i)}$ to lie within the interval $[1,1+\chi]$ for $\chi=1/\poly n$, for which by \cref{lem:noise-amplification} it suffices to choose $T=\poly n$ in such a way that the polynomial degrees of $T$ and $\chi^{-1}$ satisfy $\deg(T) > \deg(\chi^{-1})$.
By the same argument as in \cref{th:cor-1} we can further restrict the scaling constant $C$ present in \cref{eq:H-approx} to scale as $n'^\delta$, whereby the system is padded to size $n'=\poly n$.
Finally, by the definition of \cref{def:lham}, all $\| \op h_i \| = \poly n$; the resulting scaling of the simulated low-energy subspace $\Pi_-\op H'\Pi_-$ in \cref{def:approximate} is thus a polynomial, which means that the variant retains a $1/\poly n'$ promise gap.
The first three claims follow.
The last claim follows from \cref{th:main} in case $\op H_0$ was $2$-local; otherwise (which means the case $k$-local for $k>2$, as a $1$-local Hamiltonian cannot be \QMA or \QMAEXP-hard) a similar construction as in the proof of \cref{th:cor-1} for $\op H\cl^{(i)}$ can be used.
The last claim follows.
\end{proof}
\noindent
We emphasize that while the \lham problem with an exponentially small promise gap is already PSPACE-complete \cite{Fefferman2016}, the small promise gap in the reduction does not stem from an exponentially small penalty term, but because of the embedding of a PreciseQMA-hard computation. It is thus doubtful whether there is an analogue of \cref{cor:amplification} that holds for the PSPACE case.

We know there exist \QMA-hard \lham constructions with terms that all have non-varying $\BigO(1)$ weights in the system size, albeit few of them are translationally-invariant; and if they are, the local dimension is large, or the construction is contrived \cite{Gottesman2009,Bausch2016}.
\Cref{cor:amplification} is interesting for this precise reason: given a Hamiltonian with wildly-varying interaction strengths, there exists another Hamiltonian where each local term has almost zero variation in strength from site to site (apart from the two energy scales; but they apply uniformly throughout the system), and with the same hardness properties.
We thus conjecture that for any construction where translational invariance is hard to obtain, ``almost'' translational invariant models can be constructed from them, with compatible gap scaling.
This, of course, comes at the expense of changing the interaction set to allow for $\Hb$ from \cref{eq:H-bound} to be included, and modifying the interaction graph---if only by incrementing the spatial interaction topology by at most one dimension, as e.g.\ done in \cref{th:cor-1} from a two- to a three-dimensional many-body system.

As a final remark: in essence, one could achieve a similar effect as in \cref{cor:amplification} by writing a Hamiltonian $\sum_{i=1}^N \op h_i-\sum_{i=1}^N(1+\chi_i)\op h_i$.
This would be an unfair comparison though: if we expand such a Hamiltonian in a Pauli basis, there \emph{will} be small constants of $\BigO(\chi)\ll 1$; the large relative energy variations of order one are relevant for the complexity characteristics.
\Cref{cor:amplification}, on the other hand, only introduces a single, uniform energy scale, with negligible relative strength variations, even when expressed in the same Pauli basis.

\subsection{Hamiltonians with Hybrid Geared Asymptotics}
One curious feature of our construction is that it allows scaling the interaction strength of a coupling with a spatial dimension of the system at hand.
We phrase two theorems.

\begin{theorem}\label{th:cor-2}
Let $\delta>0$. There exists a translationally-invariant $2$-local Hamiltonian $\op H_{L,M}=\sum_i\op h_i$ on a square lattice of size $L\times M$ with local Hilbert space $\H$ and with open boundary conditions, for which we can define one-parameter families of Hamiltonians $S_L:=\{ \op H_{L,M(L)} \}$ and a polynomial $p(L)$, such that the following holds.
\begin{enumerate}
\item All 1- and 2-local terms either have norm 1, or norm $\Theta(L^{2+\delta})$.
\item The local spin dimension is $\le150$.
\item If $M(L)=\BigO(\log(\log(L)))$, the \lham problem for $S_L$ with promise gap $1/p(L)$ is \QMAEXP-complete.
\item If $M(L)=\Omega(\log(L))$, the \lham problem for $S_L$ is trivial for any $1/\poly$ promise gap.
\end{enumerate}
\end{theorem}
\begin{proof}
We take the \QMAEXP-complete \lham variant from \cite{Bausch2016}, which it is a translationally-invariant Hamiltonian with $1$-local interactions $\op p$ and $2$-local nearest-neighbour interactions $\op w$, each of unit norm, acting on spins with local Hilbert space $\H$ of dimension $\dim\H\le 75$, and with open boundary conditions.
Starting from this spin chain of length $L$---i.e.\ with Hilbert space $\H^{\otimes L}$---we extend it to form a square lattice of spins of side length $L\times M$ (where $M\in\field N$ is specified later), with qudits of dimension $2\dim\H$ located at the lattice vertices; we identify this new local Hilbert space with $\field C^2\otimes\H$: the extra $\field C^2$ subspace allows us to encode an extra bit of information locally at each lattice vertex.

Following a construction by \cite{Gottesman2009}, we first define the following one- and two-local interaction terms acting on neighbouring spins in the $M$ direction of the lattice---which we call a row:
\[
	\op h_1 := -\ketbra{0} \otimes \1
	\quad\text{and}\quad
	\op h_2 := \left( \ketbra{01} + \ketbra{10} \right)\otimes \1.
\]
Within each row, it is straightforward to check that these coupling terms create a unique product ground state $\ket{r_0} := \ket 0 \otimes \ket 0^{\otimes {M-1}}$; the overall ground state so far is then $\ket{r_0}^{\otimes L}$.
In the ground state, there is thus precisely one column on the lattice where all spins have flag state $\ket 0$, and all other sites across the lattice are in state $\ket 1$; and it is clear that this ground state is unique, and has a spectral gap of $1$ to the next eigenstate above it.

We now take the local interactions of the \QMAEXP-complete \lham variant, $\op w$ and $\op p$, to only act non-trivially if there is a zero flag below, i.e.\ via $\ketbra{0}\otimes\op w$, and analogously for $\op p$.
Similarly, we define a translationally-invariant bound state Hamiltonian, i.e.\ by setting $\ketbra{1}\otimes\Hb$ for $b=2$; note that for any specific column index, all the latter terms commute, and that the dimension of $\H$ ($\dim\H \ge 42$ by \cite[Th.~60]{Bausch2016}) is more than enough to implement a binary counter using only $2$-local terms (see \cite{Bausch2016}), yielding $T=2^{M-3}$ by \cref{rem:clock-times}; this includes a locally-identifiable final clock state $T$ on which we wish to condition in due course.

Now, the on-site interaction $\op p$ contains a so-called output penalty term $\op p'$, which is used to inflict an energy penalty on invalid computation outcomes; this is what pushes the ground state energy of the history state Hamiltonian up by a $1/\poly L$ amount in case of an embedded \no-instance.
We couple this penalty term to $\ketbra T$ in the biased clock Hamiltonian's space as $\op p' := \op p \otimes \ketbra T$; this term is originally $1$-local, so we do not increase the overall Hamiltonian's locality.
All other terms will remain un-coupled.

Claim 1 and 2 then follow by construction.
The consequence of scaling the system in dimension $M$ is to reduce the effect of the output penalty; by \cref{lem:H-bound-gs-2} and for $b=2$, the magnitude of the scaling will be $\propto 1/2^T$.
Since $T\propto 2^{M}$, the suppression of the error term is doubly-exponential in $M$.
It is clear that if $M=\BigO(\log(\log(L))$, the penalty term is only polynomially-suppressed, and the problem remains \QMA hard---in particular, there exists a polynomial $p$ such that
with a promise gap closing as $1/p(n)$, where $n=L\times M$ is the system size, the \yes and \no-instances of the embedded computation lie above resp.\ below the corresponding thresholds.

On the other hand, for any polynomially-closing promise gap an energy difference of $\BigO(1/\exp n)$ will essentially be invisible; for $M=\Omega(\log L)$, all embedded computational instances---irrespective of their outputs---are thus jointly either \yes or \no instances in the \lham problem.
The claim follows.
\end{proof}

It is clear that variants of this effect are easily constructed, by varying the bound state Hamiltonian $\Hb$, or by changing which terms couple to its final state.
We emphasize that the threshold $M=\BigO(\log(\log(L)))$ is, again, an arbitrary choice, and we can e.g.\ set it at $M=\BigO(\log(L))$, implying that different directions of geared limits have distinct complexity-theoretic behaviour, as the following corollary shows.
\begin{corollary}
\Cref{th:cor-2} holds also for a choice of $1/\exp n$ promise gap in the system size $n=L\times M$: the distinction then is trivial vs.\ \BQEXPSPACE-complete, for geared limits $M(L)=\BigO(\log L)$ vs $M(L)=\Omega(L)$.
\end{corollary}
\begin{proof}
	The local Hamiltonian problem with a $1/\exp n$ promise gap is \PSPACE-hard (see \cref{def:BQEXPSPACE} and the following discussion).
	In a similar fashion as for \QMA, for a translationally-invariant Hamiltonian with a single free parameter $n$---i.e.\ the system size---the input to the computational \lham problem has size $|n|=\lceil \log n \rceil$;
	this means that polynomial space in the input size $|n|$ only requires a logarithmically-sized subspace of the available spins in the many-body system.
	As a consequence, the \emph{natural} hardness for a $1/\exp n$ promise gap and the translationally-invariant \lham problem has to be \BQEXPSPACE (which can be constructed for a PreciseQMA\textsubscript{EXP} verifier, defined analogously as in \cite{Fefferman2016}).
	
	With a smaller promise gap, we need to make the extra spatial size $M(L)$ larger than before to not inflict any penalty on the computation output.
	More precisely, if $M(L)=\Omega(L)$, then the effective penalty strength is $\BigO(1/\exp(\exp n))$, much smaller than the promise gap---the problem becomes easily decidable.
	On the other hand, if $M(L) = \BigO(\log L)$, there exists a $1/\exp n$ promise gap for which the \lham variant is \BQEXPSPACE complete.
\end{proof}

\section{Conclusion}
With the construction presented in this work we show that one can significantly reduce the unphysically-large energy variations present in various models used in Hamiltonian complexity theory.
While it does not completely remove the necessity of strong interactions, it decouples the scaling from the range of the interactions and precision to be simulated; furthermore, the approximation error introduced is relative, meaning that any requirements on an error bound present in a target Hamiltonian remain intact.
This does not come for free: we need to add ancilliary qubits, potentially increase the locality and/or the local dimension of the system.
While one could certainly claim that it is arguable which model is more physical, in the end, our work draws the tradeoff between locality, local dimension, varying interaction strenght and the overall norm of an operator from a new angle.
We also emphasize that \cref{eq:H-bound} is stoquastic---i.e.\ it only has non-positive off-diagonal matrix entries.
Stoquastic \lham variants are naturally \StoqMA-complete (see discussion below the introduction of QMA in \cref{def:qma}); our construction is thus compatible with these constructions as well.

Another shift in perspective is given with regards to translational invariance.
While most if not all many-body systems in the real world have translationally-invariant interactions, many complexity-theoretic models do not.
Modfiying them to be translationally-invariant often requires an unphysically large local dimension.
With our construction, it is conceivable that non-translationally-invariant systems can be lifted to ``almost'' invariant models.
In this way, complex systems can be arbitrarily close to true translational invariance---this becomes particularly interesting if the latter are deemed easy to solve.

To draw the connection back to the abstract, creating an effective interaction strength by coupling to an extraneous Hilbert space governed by its own Hamiltonian ($\Hb$ here) is reminiscent of how e.g.\ leptons interact with each other via a bosonic exchange particle: a coupling $\bar q \gamma^\mu W_\mu q$ as e.g.\ part of the Standard Model Lagrangian, where $W_\mu$ is the boson field that determines the strength of attraction between the leptons $\bar q$ and $q$.
Similarly, the Hubbard-Stratonovich transformation, used to linearize field equations for many-body interaction terms, replaces direct particle-particle interactions into a system of equations describing independent particles coupled to a background field \cite{Stratonovich1957,Hubbard1959}.
Even though it is not meaningful to speak of the thermodynamic limit of our gadget Hamiltonian without the introduction of additional renormalization terms to cancel out emerging infinities, such methods are well-studied in the context of field theories \cite{Peskin1987,Blumenhagen,Zwiebach2004}.
It is thus reasonable to assume that similar effective perturbation gadgets as in our work can be created to introduce---and potentially explain---diverging energy scales present in many continuous theories.

There's a list of open problems we wish to address in future work.
\begin{enumerate}
\item In the commuting case, perturbation theory can be applied in parallel without the requirement of a coupling constant that scales with the system size; this also applies to our findings.
The \lham problem is not yet completely solved for the case of commuting terms, although there is progress \cite{Schuch2011,Bravyi2003,Aharonov2011}.
While our path clock in \cref{sec:H-bound} is not commuting, in \cite{Breuckmann2013} the authors present a commuting version, which could be similarly biased as ours to present a sufficient falloff.
Can one construct a commuting variant of perturbation gadgets, applicable to the commuting \lham problem?
\item What about using e.g.\ the Schrieffer-Wolff transform instead of a Feynman-Dyson series? In \cite[Sec.\ 5]{Cao2017a}, the authors have analyzed the tradeoff between the two constructions, and found the scaling to be favourable for the Schrieffer-Wolff expansion.
And can we combine our result with the numerical optimizations of the necessary scaling as in \cite{Cao2017a}?
\item Including the tiling Hamiltonian to cancel out cross terms renders \cref{eq:self-energy-full} particularly simple; indeed, if we replaced $\eta$ with $\eta^{-1}$, Catalan numbers emerge in the sum, with an emerging link to Motzkin walks \cite{Movassagh2014}. Can we learn something from this, and e.g.\ add \cref{eq:self-energy-full} up exactly?
\end{enumerate}

\section*{Acknowledgements}
We are especially grateful for helpful discussions with Stephen Piddock, Elizabeth Crosson, and Graeme Smith, and for the opportunity to present preliminary versions of this result at Caltech and CU Boulder.
We are further thankful for the useful comments obtained from the anonymous reviewers at QIP.
J.\,B.\ acknowledges support from the Draper's Research Fellowship at Pembroke College.

\printbibliography
\end{document}